\documentclass[letterpaper, 11pt]{amsart}
\usepackage[utf8]{inputenc}
\usepackage[numbers]{natbib}
\usepackage{url}
\usepackage{numprint}
\usepackage[margin=1.206in]{geometry}

\usepackage{amsmath}
\usepackage{mathrsfs}
\usepackage{amssymb}
\usepackage{amsfonts}
\usepackage[scaled=.98,sups,osf]{XCharter}
\usepackage[scaled=1.04,varqu,varl]{inconsolata}
\usepackage[type1]{cabin}
\usepackage[charter,vvarbb,scaled=1.05]{newtxmath}

\usepackage{cancel}
\usepackage[colorlinks, citecolor=citeblue, linkcolor = citeblue, urlcolor=citeblue]{hyperref}

\usepackage{tabularx}
\usepackage{multirow}
\usepackage{multicol}
\usepackage{colortbl}
\usepackage[dvipsnames,svgnames,x11names,fixpdftex,table]{xcolor}
\usepackage{numprint}
\usepackage{textcomp}
\usepackage{booktabs}
\usepackage{amsmath}
\usepackage{enumitem}
\usepackage{amssymb}
\usepackage{dsfont}
\usepackage{enumitem}
\usepackage{tcolorbox}
\usepackage{cleveref}
\DeclareMathOperator{\R}{\mathbb{R}}
\DeclareMathOperator{\N}{\mathbb{N}}
\DeclareMathOperator{\argmin}{arg\, min}
\newcommand{\e}{\mathrm{e}}

\xdefinecolor{citeblue}      {RGB}{  0,45,154}
\xdefinecolor{citered}        {RGB}{190, 69, 32}
\usepackage{graphicx}
\usepackage{wrapfig}

\usepackage{tikz}
\usepackage[framemethod=TikZ]{mdframed}
\newtheorem{proposition}{Proposition}
\newtheorem{definition}{Definition}
\newtheorem{theorem}{Theorem}
\newtheorem{lemma}{Lemma}

\newtheorem{remark}{Remark}

\newtheorem{assumption}{Assumption}
\newtheorem{example}{Example}
\renewcommand\f[1]{\mathbf{#1}}

\title{Dictionary-Sparse Recovery From Heavy-Tailed Measurements}
\author[P. Abdalla \and C. K\"ummerle]{Pedro Abdalla$^*$ \and Christian K\"ummerle$^\dagger$}
\thanks{$^*$Department of Mathematics, ETH Z{\"u}rich, Zurich, Switzerland (\href{mailto:pedro.abdallateixeira@ifor.math.ethz.ch}{pedro.abdallateixeira@ifor.math.ethz.ch})}
\thanks{$^\dagger$Department of Applied Mathematics \& Statistics, Johns Hopkins University, Baltimore, USA (\href{mailto:kuemmerle@jhu.edu}{kuemmerle@jhu.edu}).}
\date{\today}

\begin{document}%

\maketitle

\begin{abstract}
{The recovery of signals that are sparse not in a given \emph{basis}, but rather sparse with respect to an \emph{over-complete dictionary} is one of the most flexible settings in the field of compressed sensing with numerous applications. As in the standard compressed sensing setting, it is possible that the signal can be reconstructed efficiently from few, linear measurements, for example by the so-called \emph{$\ell_1$-synthesis method}. 

However, it has been less well-understood which measurement matrices provably work for this setting. Whereas in the standard setting, it has been shown that even certain heavy-tailed measurement matrices can be used in the same sample complexity regime as Gaussian matrices, comparable results are only available for the restrictive class of sub-Gaussian measurement vectors as far as the recovery of dictionary-sparse signals via $\ell_1$-synthesis is concerned.

In this work, we fill this gap and establish optimal guarantees for the recovery of vectors that are (approximately) sparse with respect to a dictionary via the $\ell_1$-synthesis method from linear, potentially noisy measurements for a large class of random measurement matrices. In particular, we show that random measurements that fulfill only a small-ball assumption and a weak moment assumption, such as random vectors with i.i.d. Student-$t$ entries with a logarithmic number of degrees of freedom, lead to comparable guarantees as (sub-)Gaussian measurements. Our results apply for a large class of both random and deterministic dictionaries.

As a technical tool, we show a bound on the expectation of the sum of squared order statistics under very general assumptions, which might be of independent interest.

As a corollary of our results, we also obtain a slight improvement on the weakest assumption on a measurement matrix with i.i.d. rows sufficient for uniform recovery in standard compressed sensing, improving on results by Lecu{\'e} and Mendelson, and Dirksen, Lecu{\'e} and Rauhut.}
\end{abstract}


\section{Introduction}
In the seminal works \cite{candes2006robust,candes2006stable,donoho2006compressed}, the idea of compressed sensing was introduced as a paradigm to recover signals from undersampled and corrupted measurements. Mathematically, the goal is to recover a signal $\mathbf{z}_0 \in \mathbb{R}^d$ from corrupted linear measurements $\f{y}=\Phi \mathbf{z}_0 + \mathbf{e} \in \mathbb{R}^m$, where $\mathbf{e}\in \mathbb{R}^m$ is an error vector satisfying $\|\mathbf{e}\|_2 \le \varepsilon$ for some constant $\varepsilon >0$ and $\Phi \in \mathbb{R}^{m \times d}$ is a sensing matrix, often also called \emph{measurement matrix}. The number of measurements $m$, which corresponds to the number of rows of $\f{\Phi}$, is assumed to be much smaller than the ambient dimension $d$, formally, $m = o(d)$. Without any additional assumption, there is not much hope to estimate $\f{z}_0$ just from knowledge of $\f{y}$ and $\f{\Phi}$.

\subsection{Recovery of Sparse Signals} \label{sec:sparseintro}
However, the theory of compressed sensing suggests that if the vector of interest $\mathbf{z}_0$ has a low-dimensional structure such that it is, for example, a \emph{sparse} vector with only few non-zero entries $s:= \|\f{z}_0\|_0 := \left|\{j \in [N] : (\f{z}_0)_j \neq 0\}\right| \ll d$, it can be stably recovered by solving a convex optimization program called \emph{quadratically constrained basis pursuit (QCBP)} or \emph{quadratically constrained $\ell_1$-minimization} \cite{foucart2013invitation}
\begin{equation}
\label{QCBP}
\hat{\mathbf{z}}= \arg \min  \|\mathbf{z}\|_1      \quad   \text{subject to}\quad \|\Phi \mathbf{z}-\mathbf{y}\|_2\leq \varepsilon
\end{equation}
such that $\hat{\mathbf{z}}$ is close to $\f{z}_0$, at least if the measurement matrix $\f{\Phi} \in \R^{m \times d}$ fulfills appropriate conditions and $m = \Omega( s \log(\e d/s))$ \cite{foucart2013invitation}. Examples for such conditions on $\f{\Phi}$ are the \emph{restricted isometry property} \cite{CT06,foucart2013invitation} or variants of the so-called \emph{null space property} \cite{cohen2009compressed}, see also \Cref{sec:dictionaryintro} below. We note that under these types of conditions, methods different from \ref{QCBP} have also been shown to exhibit similar recovery guarantees, such as greedy and thresholding-based methods \cite{ZhangOMP11,Foucart11,BouchotFoucartHitczenko16,ZhouXiuQi19,ZhaoLuo20}.

Therefore, a crucial question that lies at the core of compressed sensing is about \emph{which measurement matrices $\f{\Phi} \in \R^{m \times d}$ fulfill} such conditions in the information theoretically optimal regime of $m = \Theta( s \log(\e d/s))$ and therefore allow for robust recovery guarantees. All known constructions of measurement matrices to achieve this regime rely on some sort of randomness; it has been shown that matrices $\f{\Phi}$ with i.i.d. Gaussian entries \cite{CT06}, i.i.d. Bernoulli and other sub-Gaussian entries \cite{Baraniuk08} allow for robust guarantees with high probability.

In subsequent works, it was shown that distributions without strong concentration properties are equally suitable for the design of measurement matrices. In particular, it was shown that matrices with i.i.d. sub-exponential entries can be used with the same recovery guarantees as Gaussian matrices \cite{Koltchinskii11,Fou14}, and finally, \cite{ML17,DLR18} showed that this assumption can be relaxed to the requirement of just $\Omega(\log(d))$ small  moments for the entrywise distribution, which was shown to be necessary up to a log-log factor \cite{ML17}.

While the results on the mentioned random constructions allow for optimal sample complexity, such measurement matrices are often challenging to implement in applications. In wireless communications \cite{Hauptetal} and radar \cite{Potter10}, measurement matrices with a convolutional structure are used, and the related random partial circulant matrices have been shown to be fulfill a restricted isometry property for a number of measurements that is optimal up to logarithmic factors \cite{RRT12, KMR14,mendelson2018improved}. In magnetic resonance imaging \cite{LDSP08}, measurement matrices involving subsampled matrices of Fourier structure \cite{RV08,haviv2017restricted} or combined Fourier and wavelet structure \cite{BreakingCoherence,KW14} are of relevance, allowing for similar results.

This research activity is evidence for a desire to explore the limits in the optimal design of measurements subject to the constraints of engineering applications. It is still an open question whether an improved analysis can lead to guarantees for structured measurement matrices in the optimal sample complexity regime \emph{without} logarithmic factors.

\subsection{Recovery of Dictionary-Sparse Signals} \label{sec:dictionaryintro}
In most applications of compressed sensing, it is actually not the case that the signal of interest $\f{z}_0$ itself has few non-zero entries, but rather, the setting allows for a sparse representation such that there exists a matrix $\f{D} \in \R^{d \times n}$ and a \emph{sparse} vector $\f{x}_0 \in \R^{n}$ such that 
\begin{equation} \label{eq:signalmodel}
\f{z}_0 = \f{D} \f{x}_0.
\end{equation}
This matrix $\f{D} = \begin{bmatrix}
	\f{d}_1,&\ldots,& \f{d}_n 
\end{bmatrix}$ is called a \emph{dictionary (matrix)} \cite{Elad07}. The case mentioned above corresponds to the situation in which the dictionary $\mathbf{D} = \f{I}$ is the identity matrix.

In this paper, we focus on the case that this dictionary matrix $\f{D} \in \R^{d \times n}$ is \emph{over-complete}, i.e., is such that $n > d$, allowing for a much larger class of signals than in the standard compressed sensing setting outlined above. We will refer to signals $\f{x}_0 \in \R^{n}$ as in \ref{eq:signalmodel} with sparse $\f{z}_0 \in \R^{d}$ as \emph{dictionary}-sparse signals.

Dictionaries of interest are on the one hand pre-designed, application dependent dictionaries such as Gabor frame \cite{HermanStrohmer09}, wavelet or shearlet \cite{Guo07} dictionaries. On the other hand, it has been also fruitful to \emph{learn} over-complete dictionaries from data first to subsequently sparsely encode signals based on the obtained dictionary \cite{AharonEladBruckstein06KSVD,Wright09,RavishankarBresler13}.

As mentioned, recovering signals as given by \ref{eq:signalmodel} from few linear measurements provides a very general setup that has been studied in previous works \cite{candes2011compresseddic,aldroubi2012perturbations,ChenWangWang14,SSCoSaMP13}. Given the measurement matrix $\f{\Phi} \in \R^{m \times d}$  
and a measurement vector $\f{y}=\Phi \mathbf{z}_0 + \mathbf{e} \in \mathbb{R}^m$, with $\mathbf{e}\in \mathbb{R}^m$ and $\varepsilon >0$ as above, the \emph{$\ell_1$-synthesis method} \cite{rauhut2008compresseddic,liu2012performance} performs the signal recovery by solving
\begin{equation}
\label{l1_synthesis_method}
\begin{split}
\hspace*{1cm}\hat{\f{x}}&= \arg \min  \|\f{x}\|_1    \quad \text{subject to}   \quad   \|\Phi \mathbf{D}\mathbf{x}-\mathbf{y}\|_2\leq \varepsilon, \\
\hspace*{1cm}\hat{\f{z}}&=\mathbf{D\hat{x}},
\end{split}
\end{equation}
borrowing its name from the fact that the estimator $\hat{\f{z}} \in \R^d$ of the original signal $\f{z}_0$ is \emph{synthesized} by the dictionary $\f{D}$ and the $\ell_1$-minimizer $\hat{\f{x}}$ via the equality $\hat{\f{z}}=\mathbf{D\hat{x}}$.

Any theoretical guarantees quantifying the accuracy of $\ell_1$-synthesis \ref{l1_synthesis_method} need to depend on properties of the dictionary $\f{D} \in \R^{d \times n}$, and furthermore, on the relationship of the measurement matrix $\f{\Phi}$ and the dictionary $\f{D}$. Considering the matrix product $\f{\Phi} \f{D}$ as the measurement matrix in the setting of \Cref{sec:sparseintro} gives us one possible approach of how to obtain guarantees.

The first work with an extensive analysis of this dictionary-sparse case was provided by Rauhut, Schnass and Vanderghenyst \cite{rauhut2008compresseddic}, who showed that if the dictionary $\f{D}$ fulfills the restricted isometry property (RIP) and if $\f{\Phi} \in \R^{m \times d}$ is a matrix with independent sub-Gaussians rows with $m = \Omega(s \log(\e n/s))$, the product $\f{\Phi} \f{D}$ also fulfills the RIP of sparsity order $s$ with high probability. In \cite{ChenWangWang14}, Chen, Wang and Wang introduced a null space property tailored to the dictionary case, which is defined as follows: 
If $\f{D} \in \R^{d \times n}$ is a dictionary, it is said that $\Phi \in \mathbb{R}^{m \times d}$ satisfies the \emph{$\mathbf{D}$-null space property ($\f{D}$-NSP)} of order $s$ if for any index set $T$ with $|T| \le s$ and any vector $\mathbf{v}$ such that $\mathbf{Dv} \in \operatorname{ker}\f{\Phi}  \setminus \{0\}$, there exists some $\mathbf{u} \in \operatorname{ker} \mathbf{D}$ such that $\|\mathbf{v}_T + \mathbf{u}\|_1 < \|\mathbf{v}_{T^c}\|_1$. \cite{ChenWangWang14} showed that the $\f{D}$-NSP is a necessary and sufficient condition for the exact recovery of dictionary-sparse signals via $\ell_1$-synthesis method \ref{l1_synthesis_method} in the noiseless case where $\epsilon = 0$.

Using this notion of $\f{D}$-NSP, the authors of \cite{ChenWangWang14} were also able to relate exact recovery guarantees of \ref{l1_synthesis_method} with a null space property of $\f{\Phi} \f{D}$ in the conventional sense of \cite{foucart2013invitation}, which we recall as follows.
\begin{definition}[{\cite[Chapter 4]{foucart2013invitation}}] \label{def:NSP}
Let $\f{A} \in \R^{m \times n}$ be a matrix and $s \in \N$.
\begin{enumerate}
    \item $\f{A}$ is said to fulfill the \emph{null space property (NSP)} of order $s$ if for all $\f{v} \in \ker(\f{A}) \setminus \{0\}$ and all subsets $T \subset [n]$ with $|T|\leq s$, 
    \[
    \|\f{v}_{T}\|_1 < \|\f{v}_{T^c}\|_1.
    \]
    \item $\f{A}$ is said to fulfill the \emph{robust null space property (robust NSP)}\footnote{See also \Cref{remark:nsps}.} of order $s$ with constants $0 < \gamma < 1$ and $\tau > 0$ if
    \begin{equation*}
    \inf_{\f{v} \in S_{\gamma}} \|\mathbf{A}\mathbf{v}\|_2 \ge \frac{1}{\tau},
	\end{equation*}
	where $S_{\gamma} := \left\{ \f{v} \in \R^n: \|\f{v}_T\|_2 \geq \frac{\gamma}{\sqrt{s}} \|\f{v}_{T^c}\|_1 \text{ for some } T \subset [n] \text{ with }|T|=s\right\} \cap \mathbb{S}^{n-1}$.
\end{enumerate}
\end{definition}
While null space property is a necessary and sufficient condition for the successful recovery of sparse vectors via \ref{QCBP} in the noiseless case of $\varepsilon = 0$, the robust null space property implies also robust recovery under noisy measurements and model error. \cite{ChenWangWang14} obtained the following result.
\begin{proposition}[{\cite[Theorem 7.2]{ChenWangWang14}}] \label{prop:ChenWangWang}
If the dictionary $\mathbf{D} \in \R^{d \times n}$ is full spark, i.e., if no set of $d$ columns of $\f{D}$ is linear dependent, then $\f{\Phi}$ satisfies the $\mathbf{D}$-NSP of order $s$ if and only if $\f{\Phi} \mathbf{D}$ satisfies the NSP of the same order.
\end{proposition}
We note that the full spark assumption for the dictionary $\f{D}$ is not a strong assumption. Indeed, if $\f{D}$ is a matrix whose independent columns are drawn from a continuous distribution, then it has full spark with probability one \cite{alexeev2012full}. 

We now turn to the core questions addressed in this paper. Taking Proposition \ref{prop:ChenWangWang} as a starting point, Casazza, Chen and Lynch proved recently in \cite{CasazzaChenLynch20} {the following: They showed that} if the dictionary $\mathbf{D}$ is a deterministic matrix {of full spark} with columns bounded in $\ell_2$-norm that furthermore fulfills the robust NSP of {Definition} \ref{def:NSP}, a measurement matrix $\f{\Phi} {\in \R^{m \times d}}$ with \emph{i.i.d. sub-Gaussian rows} {satisfying $m = \Omega(s \log(\e n/s))$} results in a product $\f{\Phi} \f{D}$ that fulfills the robust NSP {with high probability}. Via this robust NSP, their results ensure that for such dictionary and measurements, the $\ell_1$-synthesis method \ref{l1_synthesis_method} recovers $s$-dictionary sparse signals $\f{z}_0 \in \R^{d}$ with high probability, robustly from noisy measurements if $m = \Omega(s \log(\e n/s))$.

However, while being an improvement over the result of \cite{rauhut2008compresseddic} as the NSP assumption on $\f{D}$ is weaker than the RIP of \cite{rauhut2008compresseddic}, the sub-Gaussianity assumption of \cite{CasazzaChenLynch20} on the measurements $\f{\Phi}$ is still quite restrictive. This becomes in particular clear when comparing such an assumption to the variety of results on heavy-tailed measurement matrices for the conventional compressed sensing setting detailed in \Cref{sec:sparseintro}. In this paper, we therefore address the following questions:

\emph{\begin{enumerate}
    \item[(Q1)] What are the weakest assumptions for suitable measurement matrices $\f{\Phi}$ for uniform recovery guarantees of dictionary-sparse signals via $\ell_1$-synthesis?  \label{Q1}
    \item[(Q2)] Which joint assumptions on $\f{\Phi}$ and the dictionary matrix $\f{D}$ come with good trade-offs? \label{Q2}
\end{enumerate}}

\subsection*{Our Contribution}
We provide results for both deterministic and random dictionaries $\f{D}$, addressing \emph{(Q1)} and \emph{(Q2)}. In particular, we show in \Cref{main_result_deterinistic_dictionary,main_corollary_deterministic_dictionary} that for the case of a full spark dictionary $\f{D}$ fulfilling a robust NSP, $\f{\Phi} \f{D}$ fulfills the robust NSP if $\f{\Phi}$ is a measurement matrix with \emph{i.i.d. rows whose distributions only have a logarithmic number of finite, well-behaved moments (and also fulfill a small-ball condition)}, with high probability in the optimal regime of $m = \Omega(s \log(\e n/s))$. 

Secondly, we prove two results about the setting where both $\f{\Phi}$ and $\f{D}$ obey a random model. Via \Cref{main_result_first_random_dictionary}, we show that both $\f{\Phi}$ and $\f{D}$ can obey heavy-tailed distributions while still retaining the robust NSP of $\f{\Phi} \f{D}$. Finally, we show with \Cref{main_result_second_random_dictionary} that if the dictionary $\f{D} \in \R^{d \times n}$ has i.i.d. entries whose distribution has $\log(n)$ sub-Gaussian moments, it is sufficient to require only a bounded variance and a small-ball condition for the entries of the measurement matrix $\f{\Phi}$. The established property of $\f{\Phi} \f{D}$ for each of the mentioned models on $\f{\Phi}$ and $\f{D}$ implies that with high probability, the $\ell_1$-synthesis method is able to uniformly {recover} all dictionary-sparse signals of order $s$, stably and robustly under measurement noise.

Lastly, we prove a lemma on order statistics that bounds the expectation of the sum of the largest coordinates of a random vector with just few sub-Gaussian moments in \Cref{theorem:mendelson}, generalizing a result of \cite{MendelsonLearning15}. This result is key for several of our results, and furthermore, improves on \cite{ML17} and \cite{DLR18} by relaxing the weakest known moment assumption for standard compressed sensing in the optimal regime. This aspect is presented and discussed in \Cref{D_satisfies_NSP_whp}.

The structure of this paper is organized as follows. After reviewing related works, we provide some preliminaries in \Cref{sec:background}. In \Cref{sec:results:deterministic}, we present our results for deterministic dictionary matrices $\f{D}$ and random measurements $\f{\Phi}$ and at the end we introduce our main technical result. In \Cref{sec:results:random} we present our results for joint random models on the dictionary $\f{D}$ and the measurement matrix $\f{\Phi}$ in \Cref{sec:results:random}. Finally, we present the proofs of our results in \Cref{sec:proofs}, and conclude with a short discussion in \Cref{sec:conclusion}.

\subsection{Related Works}
As mentioned above, results of similar flavor as ours can be found in \cite{rauhut2008compresseddic,ChenWangWang14,CasazzaChenLynch20}, albeit subject to considerably stronger assumptions on the measurement and dictionary matrices. 

Going beyond \emph{uniform} results, the work \cite{GiryesElad13} observes that under quite weak assumptions on the dictionary $\f{D}$, \emph{signal} recovery might still be possible for $\ell_1$-synthesis \ref{l1_synthesis_method} even when the coefficients are not uniquely recoverable. This idea has been extended in the recent work \cite{Maerz20} by considering a \emph{non-uniform}, i.e., signal-dependent analysis based on conic Gaussian mean widths. Supported by numerical evidence, the authors of \cite{Maerz20} provide a characterization which signals might be uniquely recoverable via $\ell_1$-synthesis despite a non-unique sparse coefficient representation. 

A different route for the recovery of signals that are sparse in an over-complete dictionary is taken in \cite{SSCoSaMP13,GiryesNeedell15}, where a compressive sampling matching pursuit (CoSaMP) is proposed and analyzed. This method contains a projection step that is in general not efficently implementable, and furthermore, the presented theory requires the measurement matrix $\f{\Phi}$ to fulfill a RIP, which is stronger than the requirements of \cite{CasazzaChenLynch20}.

An approach that is related to $\ell_1$-synthesis, but significantly different, is the $\ell_1$-analysis model, which as been studied in an arguably even larger body of literature, cf. \cite{Elad07,candes2011compresseddic,NamDavies13,Genzel20} and the references therein. Instead of recovering signals $\f{z}_0$ as in \ref{eq:signalmodel}, this model assumes that $z_0 \in \R^d$ is such that $\f{D}^* \f{z}_0$ is sparse (which is not equivalent to our model for non-square dictionaries $\f{D} \in \R^{d \times n}$) and solves the convex program
\[
\hspace*{1cm}\hat{\f{z}}= \arg \min_{\f{z} \in \R^d}  \|\f{D}^*\f{z}\|_1    \quad \text{subject to}   \quad   \|\Phi \mathbf{z}-\mathbf{y}\|_2\leq \varepsilon.
\]
In this context, the notion of a $\f{D}$-restricted isometry property ($\f{D}$-RIP) was shown to be helpful \cite{candes2011compresseddic}, and it was shown to hold for random measurement matrices fulfilling strong concentration property of enough measurements are provided \cite{candes2011compresseddic} which is stronger than the $\f{D}$-NSP mentioned above. These results were extended to measurements constructed from subsampled bounded orthogonal systems \cite{KrahmerNeedellWard15}, which are of relevance in magnetic resonance imaging.

Both the $\ell_1$-synthesis and the $\ell_1$-analysis model come with their own set of advantages and limitations. For extensive comparisons and discussion, we refer to \cite{Elad07,liu2012performance,NamDavies13}.

\section{Preliminaries and Background} \label{sec:background}
To begin with, we introduce some notation. In a mathematical context, {we write "$\e$" for the natural number}, $\|\cdot\|_p$ for the standard $\ell_p$-norm and $\mathbb{S}_p^{n-1}$, $\mathbb{B}_p^n$ for the unit sphere and unit ball in $\mathbb{R}^n$ with respect to the $\ell_p$ norm, respectively. {When the sub-index is omitted, it is assumed to be $p=2$}. Moreover, for every matrix $\f{A}$ we denote the standard operator norm by $\|\f{A}\|$. To avoid confusion, for a random variable $X$, we denote its $L_p$-norm by $\|X\|_{L^p} = \mathbb{E}[|X|^p]^{1/p}$ for $1 \le p < \infty$, and write the essential supremum of $X$ as $\|X\|_{L^{\infty}}$. For sets, for every natural number $n$, definite $[n]:=\{1,\ldots,n\}$.{We also denote the complement of the set $T$ by $T^{c}$}. For a vector $x = (x_1,\ldots,x_n) \in \mathbb{R}^n$, the non-increasing rearrangement $x^{\ast}$ of $x$ is the vector $x^{\ast}=(x_1^{\ast},\ldots,x_n^{\ast})$ such that its coordinates are listed in a non-increasing order with respect to their magnitudes, i.e, $|x_1^{\ast}|\ge |x_2^{\ast}|\ge \ldots \ge |x_n^{\ast}| $.
For functions $f(m,s,n)$, $g(m,s,n)$ of $m$, $s$ and $n$ that are (varying) parameters, we write $f \lesssim g$, $f=O(g)$ or $g=\Omega(f)$ if there exists an absolute constant $C>0$ such that $ f(m,s,n) \leq C g(m,s,n)$. In a similar fashion, $f=\Theta(g)$ if simultaneously $f=O(g)$ and $g=O(f)$. Finally, $\sigma_s(\f{x})_{1}$ denotes the $\ell_1$\emph{-error of  the best $s$-term approximation} of a vector $\f{x} \in \mathbb{R}^n$, i.e., $\sigma_s(\f{x})_{1}=\inf \{\|\f{x}-\f{z}\|_1: \ \f{z} \in \mathbb{R}^n \ \text{is $s$-sparse} \}$.

An important question in the context of the recovery of signals that have $s$-sparse coefficients with respect to an over-complete dictionary $\mathbf{D}  = [\mathbf{d}_1,\ldots,\mathbf{d}_n] \in \mathbb{R}^{d\times n}$ with $d \ll n$ is \emph{which dictionaries} are admissible for successful recovery of $\f{z}_0$ via $\ell_1$-synthesis. The result of \cite{rauhut2008compresseddic} allows for dictionaries $\f{D}^{d \times n}$ that fulfill a restricted isometry property (RIP) of order $s$, which is for example, fulfilled for \emph{incoherent} dictionaries, i.e., those that fulfill the estimate
\[
\mu(\f{D}):= \max_{i \neq j} \frac{|\langle \f{d}_i,\f{d}_j \rangle|}{\|\f{d}_i\|_2 \|\f{d}_j\|_2} \leq \frac{1}{16 (s-1)}.
\]
Using this incoherence estimate, it is possible to find deterministic, admissible dictionaries. However, it is well-known \cite[Section 5.2]{foucart2013invitation} that this allows only for small {sparsity levels} of $s \lesssim \sqrt{d}$. Random constructions for $\f{D}$ using i.i.d. sub-Gaussian entries, on the other hand, satisfy a restricted isometry property of order $s$ for values of $s$ scaling almost linearly with $d$. 

As observed by \cite{ChenWangWang14,CasazzaChenLynch20}, it is possible to relax the RIP assumption on $\f{D}$ to a weaker null space property as defined in {Definition} \ref{def:NSP}. Together with a suitable measurement matrix $\f{\Phi}$, such a dictionary will still enable stable recovery of coefficient vector $\f{x}_0$ via $\ell_1$-synthesis \ref{l1_synthesis_method}, using the following standard result (see also \cite[Theorem 1.1]{CasazzaChenLynch20} for {another} version).
\begin{proposition}[{\cite[Theorem 4.25]{foucart2013invitation}}] \label{prop:recguarantee}
 Suppose that a matrix $\f{A} \in \R^{m \times n}$ satisfies the robust NSP of order $s$ of {Definition} \ref{def:NSP} with constants $0< \gamma < 1$ and $\tau >0$. Then, for any $\f{x}_0 \in \R^{n}$, the solution $\f{\hat{x}}= \argmin_{\|\f{A} \f{x} - \f{y}\|_2 \leq \varepsilon} \|\f{x}\|_1$ with $\f{y} = \f{A} \f{x}_0 + \f{e}$ and $\|\f{e}\|_2 \leq \varepsilon$ approximates the vector $\f{x}_0$ with {$\ell_2$-error}
 \[
 \|\f{x}_0 - \f{\hat{x}}\|_2 \leq \frac{2(1+\gamma)^2}{(1- \gamma)}\frac{\sigma_s(\f{x}_0)_1}{\sqrt{s}} + \frac{4 \tau}{1- \gamma} \varepsilon,
 \]
 where $\sigma_s(\f{x}_0)_1$ is the best $s$-term $\ell_1$-approximation error of $\f{x}_0$.
\end{proposition} 

Using this proposition for $\f{A} = \f{\Phi}\f{D}$, it is possible to derive guarantees for stable recovery of the coefficient vector $\f{x}_0$ via $\ell_1$-synthesis \ref{l1_synthesis_method} in case that $\f{\Phi} \f{D}$ fulfills a robust NSP. In {Proposition} \ref{prop:ChenWangWang}, it was shown that for dictionaries $\f{D}$ with \emph{full spark}, the just slightly weaker NSP of $\f{\Phi} \f{D}$ is basically necessary for the success of $\ell_1$-synthesis. 

Based on this argument, it has been noted in \cite{CasazzaChenLynch20} that actually, it is further necessary that $\f{D}$ fulfills a null space property, too. This is due to the fact that the null space of $\f{D}$ is contained in the null space of $\f{\Phi}\f{D}$, i.e., $\text{Ker} \ \mathbf{D} \subset \text{Ker} \ \Phi\mathbf{D}$.

Building on this, Casazza, Chen and Lynch obtained the following result.
\begin{proposition}[{\cite[Theorem 3.3 and Corollary 3.6]{CasazzaChenLynch20}}] \label{prop:mainresult:Casazza}
	Suppose that $\f{D} \in \R^{d \times n}$ has the robust NSP of order $s$ with constant $0 < \gamma < 1$ and $\tau > 0$ and its columns satisfy $\max \{ \|\f{d}_i\|_2^2: i \in [n]\} \leq \rho$.
	
	Let the measurement matrix $\f{\Phi} = \begin{bmatrix}
		\f{\varphi}_1,&\ldots,& \f{\varphi}_m
	\end{bmatrix}^T$ have independent rows that are distributed as a centered, sub-isotropic, sub-Gaussian vector $\f{\varphi}$, i.e., $\f{\varphi}$ fulfills $\mathbb{E}[\f{\varphi}] = \f{0}$, and there exist $c > 0$ and $\sigma > 0$ such that
	\begin{enumerate}
		\item $\mathbb{E}\left[| \langle \f{\varphi}, \f{z} \rangle |\right] \geq c$ for all $\f{z} \in \mathbb{S}_2^{d-1}$,
		\item $\mathbb{P}\left(| \langle \f{\varphi}, \f{z} \rangle | \geq t \right) \leq 2 \exp(-t^2/(2 \sigma^2))$ for all $\f{z} \in \mathbb{S}_2^{d-1}$.
	\end{enumerate} 
	If the number of measurements $m$ satisfies 
	\begin{equation*}
	m \gtrsim \frac{\sigma^6}{c^6} \frac{\tau^2 \rho}{\gamma^2} s \log\left(\frac{n}{s} \right), 
	\end{equation*}
	then with probability at least $1- \exp(-m \frac{c^4}{64^2 \sigma^4})$, $\f{\Phi} \f{D}$ fulfills the robust NSP of order $s$ with constants $\rho$ and $\tau/\sigma$. In this case, $\ell_1$-synthesis \eqref{l1_synthesis_method} provides stable and robust recovery of both the coefficient vector $\f{x}_0 \in \R^{n}$ and the signal $\f{z}_0 \in \R^{d}$ from $y=\Phi \mathbf{z}_0 +e$ with $\|e\|_2 \le \varepsilon$ such that
 \begin{equation} \label{eq:Casazza:guarantees}
 \begin{split}
 \|\f{\hat{x}} - \f{x}_0\|_2 &\lesssim \sigma_s(\f{x}_0)_1 + \frac{\tau}{\sigma} \varepsilon \quad \text{ and} \\
  \|\f{\hat{z}} - \f{z}_0\|_2 &\lesssim \|\f{D}\|\left( \sigma_s(\f{x}_0)_1+ \frac{\tau}{\sigma}\varepsilon\right).
  \end{split}
 \end{equation}
\end{proposition}
{
\begin{remark} \label{remark:nsps}
In the paper \cite{CasazzaChenLynch20}, the \emph{robust NSP} of {Definition} \ref{def:NSP} is actually called \emph{stable NSP}. We chose our terminology since {Definition} \ref{def:NSP} is very similar to what is called $\ell_2$-robust NSP in \cite[Chapter 4.3]{foucart2013invitation} and \cite{DLR18}. More specifically, it is a property that is shown in the proofs of \cite{DLR18} to imply what is called the $\ell_2$-robust NSP in \cite{DLR18}.

Furthermore, we note that the approach of \cite{CasazzaChenLynch20} to derive the recovery guarantees \eqref{eq:Casazza:guarantees} under the noisy setting leads to a slightly suboptimal guarantees. To see this, compare \cite[Theorem 1.1]{CasazzaChenLynch20} to \cite[Theorem 4.22]{foucart2013invitation} (and to {Proposition \ref{prop:recguarantee}}, for that matter). For this reason our recovery guarantees presented in this paper, such as Theorem \ref{main_result_deterinistic_dictionary}, are sharper than the ones presented in \cite{CasazzaChenLynch20}.

\end{remark}
}

{Proposition} \ref{prop:mainresult:Casazza} is an improvement over the result of \cite{rauhut2008compresseddic} as its requirement on the dictionary $\f{D}$ is weaker than the restricted isometry (RIP) assumption of \cite{rauhut2008compresseddic}. We refer to \cite{CahillGap16,DLR18} for a discussion of the gap between NSP and RIP conditions.

While the robust NSP requirement on the dictionary $\f{D}$ is deterministic in its nature, it is well-known that certifying such a condition for a given matrix $\f{D}$ is in general NP-hard \cite{Tillmann13}. Typically, this issue is addressed by using random constructions, see \cite{foucart2013invitation,ML17,DLR18} for existing results. However, the existing theory does not address the intricacy of the $\ell_1$-synthesis model, as random constructions for both $\f{\Phi}$ and $\f{D}$ might \emph{not} directly imply a robust NSP for the product matrix $\f{\Phi} \f{D}$.

\section{Deterministic Dictionaries} \label{sec:results:deterministic}
In this section, we state our results for $\ell_1$-synthesis in the case of deterministic dictionaries $\f{D} \in \R^{d \times n}$, which generalize {Proposition} \ref{prop:mainresult:Casazza} of \cite{CasazzaChenLynch20} directly to a much larger class of random measurement matrices $\f{\Phi} \in \R^{m \times d}$. 

\subsection{Main Results for Deterministic Dictionaries}
In particular, we assume the following for the measurement matrix $\Phi \in \mathbb{R}^{m\times d}$ in \Cref{main_result_deterinistic_dictionary} below.
\begin{assumption}[$\f{\Phi}$ with independent rows with $\log(\e n/s)$ well-behaved moments] \label{assumption:1}
Let $\Phi \in \R^{m \times d}$ be a matrix with i.i.d. rows distributed as $\varphi/\sqrt{m} \in \mathbb{R}^d$, where $\varphi$ satisfies the conditions
\begin{enumerate}
\label{conditions_varphi}
    \item  $\mathbb{E}[\f{\varphi}]=\f{0}$,
    \item  There exists constants $K,c>0$ such that $\inf_{\f{x}\in \mathbb{S}^{d-1}}\mathbb{P}\big(|\langle \f{\varphi} , \f{x} \rangle| \ge K\big) \ge c$,
    \item  If $n \in \N$ denotes the number of columns of $\f{D}$ and $s \in \N$ arbitrary, there exist constants $\alpha \geq \frac{1}{2}$ and $\lambda > 0$ such that for every vector $\f{a} \in \mathbb{S}^{d-1}$, $\|\langle \f{\varphi},\f{a}\rangle\|_{L^p} \le \lambda p^{\alpha}$ for all $2\le p \lesssim \log \frac{\e n}{s}$.
\end{enumerate}
\end{assumption}

The first condition of \Cref{assumption:1} centers the distribution at $\f{0}$ and is not restrictive. The second condition (2) is often called \emph{small-ball assumption} and is related to an important technical component in our analysis shared by \cite{MendelsonLearning15,ML17,koltchinskii2015bounding}. It is a much weaker assumption than the sub-Gaussian assumption of {Proposition} \ref{prop:mainresult:Casazza} if we assume some mild normalization. In particular, if an isotropic random vector $\varphi$ satisfies $\|\langle \f{\varphi},\f{a}\rangle\|_{L^{2+\varepsilon}} = O(\|\langle \f{\varphi},\f{a}\rangle\|_{L^2})$ for some $\varepsilon >0$, then a simple Paley-Zygmund argument (cf. \cite[Lemma 4.1]{MendelsonLearning15}) shows that the small-ball assumption is satisfied.

Finally, condition (3) of \Cref{assumption:1} allows for distributions with much heavier tail behavior than just sub-Gaussian distributions; it is well-known that if a vector $\f{\varphi}$ is sub-Gaussian and $\f{a} \in \mathbb{R}^d$, then $\|\langle \f{\varphi},\f{a}\rangle\|_{L^p} = O(\sqrt{p}\|\langle \f{\varphi},\f{a}\rangle\|_{L^2})$ (see \cite{vershynin2018high} for a proof), i.e., the moment bound of (3) is fulfilled with $\alpha = \frac{1}{2}$ \emph{for all moments $p \in \N$}. On the other hand, (3) requires only the first $\log(\e n/s)$ moments to follow such behavior. Furthermore, the condition also includes random vectors whose first $\log(\e n/s)$ moments follow the behavior of an exponential (for $\alpha = 1$) random variable or even heavier ones (for $\alpha > 1$). We refer the reader to \cite{DLR18} for more details and examples of similar distributions.

\begin{remark}
In the following, since the parameters $K$, $c$, $\lambda$ and $\alpha$ are constant for any given distribution, they are considered as constant for the purposes of any big-$O$ type notation and thus omitted in such a notation, and similarly, in expressions using $\gtrsim$ and $\lesssim$.
\end{remark}

For measurement matrices $\f{\Phi}$ fulfilling \Cref{assumption:1}, we obtain the following theorem. 

\begin{theorem}
\label{main_result_deterinistic_dictionary}
Suppose that $\f{D} = [\mathbf{d}_1,\ldots,\mathbf{d}_n] \in \R^{d \times n}$ has the robust NSP of order $s$ with constant $0 < \gamma < 1$ and $\tau > 0$ and that its columns satisfy $\max \{ \|\f{d}_i\|_2: i \in [n]\} \leq \rho = \Theta (1)$.

Assume that the measurement matrix $\f{\Phi} \in \mathbb{R}^{m \times d}$ satisfies \Cref{assumption:1} with moment growth parameter {$\alpha \ge \frac{1}{2}$}.	
If the number of measurements satisfies
\begin{equation*}
   { m = \Omega\left(\max\left\{s\log\left(\frac{\e n}{s}\right), \log\left(\frac{\e n}{s}\right)^{\max(2\alpha-1,1)}\right\}\right)},
\end{equation*}
then with {probability at least} $1-e^{-\Omega(m)}$,  $\f{\Phi} \f{D}$ fulfills the robust NSP of order $s$ with some constants $0 < \gamma < 1$ and $\widetilde{\tau}=\Theta(\tau) > 0$.\footnote{$\widetilde{\tau}$ only depends on the distribution parameters of $\varphi$.} Thus, in this case, $\ell_1$-synthesis \cref{l1_synthesis_method} provides stable and robust recovery of any coefficient vector $\f{x}_0 \in \R^{n}$ and corresponding signal $\f{z}_0 =\f{D}\f{x}_0 \in \R^{d}$ from $y=\f{\Phi} \mathbf{z}_0 +\f{e}$ with $\|\f{e}\|_2 \le \varepsilon$ such that
 \[
 \|\f{\hat{x}} - \f{x}_0\|_2 \lesssim \frac{\sigma_s(\f{x}_0)_1}{\sqrt{s}} + \tau \varepsilon \quad \text{ and}
 \]
 \[
  \|\f{\hat{z}} - \f{z}_0\|_2 \lesssim \|\f{D}\|\left( \frac{\sigma_s(\f{x}_0)_1}{\sqrt{s}} + \tau \varepsilon\right).
 \]
\end{theorem}

\begin{remark}
	Compared to the result of \cite{CasazzaChenLynch20} which we presented in {Proposition} \ref{prop:mainresult:Casazza}, we obtain further a {scaling in the best} $s$-term $\ell_1$-approximation error $\sigma_s(\f{x}_0)_1$ improved by $1/\sqrt{s}$, which is in line with the optimal results of standard compressed sensing, cf. \cite[Chapter 4.3]{foucart2013invitation}.	It is well-known (e.g., \cite[Theorem 4.22]{foucart2013invitation}) how to derive guarantees on the $\ell_q$-error for $1 \leq q \leq 2$ based on such a property, we omit them for simplicity. 
\end{remark}

Our proof of \Cref{main_result_deterinistic_dictionary} is based on the ideas of the small-ball method \cite{koltchinskii2015bounding,MendelsonLearning15,ML17,DLR18,CasazzaChenLynch20}, and proceeds by establishing lower bounds {on a} certain empirical process. {As in \cite{MendelsonLearning15,ML17,DLR18}, but crucially, unlike \cite{CasazzaChenLynch20}}, we do \emph{not} rely on Gaussian widths in our argument, as this would rule out using assumptions as general as in \Cref{assumption:1}. {It further relies on a bound on the sum of squared order statistics, cf. \Cref{theorem:mendelson} below.} We refer to \Cref{sec:proofs:deterministic} for the details. 

{
We also point out that, although the small ball method combined with refined probabilistic techniques leads to sharp results under weak assumptions as in Theorem \ref{main_result_deterinistic_dictionary}, it is not applicable to some important examples such as measurement matrices/dictionaries formed by Fourier or wavelet basis. This is because such examples violate the small ball assumption, see \cite{ML17} for more details and related open problems.}
Next, we provide a result that it is almost a corollary of \Cref{main_result_deterinistic_dictionary}, as it is not quite a corollary, but follows from a modification of the proof.
{It is based on the following set of assumptions.}

\begin{assumption}[Measurement matrices with $\log(\e n/s)$ well-behaved moments] \label{assumption:2}
Let $\f{\Phi} \in \mathbb{R}^{m\times d}$ be a measurement matrix with i.i.d. rows distributed as $\varphi/\sqrt{m} \in \mathbb{R}^d$, where $\varphi = \begin{bmatrix}
	\xi_1, & \ldots, & \xi_d
\end{bmatrix}$ is a random vector with independent entries. Assume that for each $i \in [d]$, $\xi_i$ satisfies the conditions
\begin{enumerate}
\label{conditions_varphi_i.i.d_entries}
    \item  $\mathbb{E}[\xi_i]=0$,
    \item  $\mathbb{E}[\xi_i^2] = 1$ (i.e., $\xi_i$ has unit variance),
    \item  If $n \in \N$ denotes the number of columns of $\f{D}$ and $s \in \N$ arbitrary, there exist constants $\alpha \geq \frac{1}{2}$ and $\lambda > 0$ such that $\|\xi_i\|_{L^p} \le \lambda p^{\alpha}$ for all $2\le p \lesssim \log \frac{\e n}{s}$.
\end{enumerate}
\end{assumption}

We remark that the motivation to consider  \Cref{assumption:1} is to make it as minimal as possible, in the sense that it includes a wide range of distributions. On the other hand, \Cref{assumption:2} is easy to check. The price of being simpler is that \Cref{assumption:2} is less general and it excludes natural distributions included in \Cref{assumption:1}. For example, a random vector uniformly distributed in the Euclidean sphere with radius $\sqrt{d}$ does not have independent coordinates and it is not hard to check that \Cref{assumption:1} is satisfied as such random vector has continuous distribution and it is subgaussian, see \cite[Theorem 3.4.6]{vershynin2018high}. Now, we state our theorem adapted to \Cref{assumption:2}.
\begin{theorem}
\label{main_corollary_deterministic_dictionary}
Suppose that $\f{D} = [\mathbf{d}_1,\ldots,\mathbf{d}_n] \in \R^{d \times n}$ has the robust NSP of order $s$ with constant $0 < \gamma < 1$ and $\tau > 0$ and that its columns satisfy $\max \{ \|\f{d}_i\|_2: i \in [n]\} \leq \rho = \Theta (1)$.

Assume that the measurement matrix $\f{\Phi} \in \mathbb{R}^{m \times d}$ satisfies \Cref{assumption:2} with moment growth parameter {$\alpha \ge \frac{1}{2}$}.	
If the number of measurements satisfies
\begin{equation*}
   { m = \Omega\left(\max\left\{s\log\left(\frac{\e n}{s}\right), \log\left(\frac{\e n}{s}\right)^{\max(2\alpha-1,1)}\right\}\right)},
\end{equation*}
then with probability of at least $1-e^{-\Omega(m)}$,  $\f{\Phi} \f{D}$ fulfills the robust NSP of order $s$ with some constants $0 < \gamma < 1$ and $\widetilde{\tau} > 0$. Thus, in this case, $\ell_1$-synthesis \cref{l1_synthesis_method} provides stable and robust recovery of any coefficient vector $\f{x}_0 \in \R^{n}$ and corresponding signal $\f{z}_0 = \f{D} \f{x}_0 \in \R^{d}$ from $y=\f{\Phi} \mathbf{z}_0 +\f{e}$ with $\|\f{e}\|_2 \le \varepsilon$ such that
 \[
 \|\f{\hat{x}} - \f{x}_0\|_2 \lesssim \frac{\sigma_s(\f{x}_0)_1}{\sqrt{s}}  + \tau \varepsilon \quad \text{ and}
 \]
 \[
  \|\f{\hat{z}} - \f{z}_0\|_2 \lesssim \|\f{D}\|\left( \frac{\sigma_s(\f{x}_0)_1}{\sqrt{s}}  + \tau \varepsilon\right).
 \]
\end{theorem}

{The proof of \Cref{main_corollary_deterministic_dictionary} is presented in \Cref{sec:proofs:deterministic}.} Even in the basis case that the dictionary $\f{D} = \f{I}$ is the identity matrix (and thus, $d=n$), \Cref{main_corollary_deterministic_dictionary} slightly improves \cite[Corollary 8]{DLR18} because it requires fewer moments of $\varphi$ and do not require the entries to be identically distributed: indeed, \cite[Corollary 8]{DLR18} requires the bound of condition (3) of \Cref{assumption:2} for the first $\log(n)$ moments, whereas we only require that bound for the first $\log(\e n/s)$ moments.

\begin{example}
Measurement matrices fulfilling \Cref{assumption:2} include many random ensembles with heavy-tailed distributions. For example, let $\xi_1,\ldots,\xi_d$ be independent Student-$t$ random variables of degree $k$. Student-$t$ variables are neither sub-Gaussian nor sub-exponential and even do not have finite $r$-th moment if $r > k$ \cite{Kotz04}. On the other hand, the moments grow as the ones of a sub-Gaussian random variable for $r \in \{1,\ldots k/2\}$.

Therefore, \Cref{main_corollary_deterministic_dictionary} implies that a measurement matrix $\f{\Phi}$ with independent rows whose entries are i.i.d. Student-$t$ random variables of degree $k = \Omega(\log(\e n/s))$ is suitable for robust recovery of $s$-dictionary-sparse signals via $\ell_1$-synthesis. We refer to \cite[Example 9]{DLR18} for more examples for distributions fulfilling similar assumptions.
\end{example}

\subsection{Main technical result: Bound on Expectation of Sum of Squared Order Statistics}
{One crucial step towards proving \Cref{main_result_deterinistic_dictionary,main_corollary_deterministic_dictionary} is a result about the expectation of the sum of squared order statistics. Before we formally state the result, recall that for random variables $z_1,\ldots,z_n$, we call the $i$-th coordinate of the non-increasing rearrangement $z_i^*$ by the $i$-th \emph{order statistic} of the set $(z_1,\ldots,z_n)$ \cite{DavidNagaraja04}. We expect the following theorem to be of independent interest in statistics and in the theory of sparse recovery.}

\begin{theorem}[{Generalization of \cite[Lemma 6.5]{MendelsonLearning15}}] \label{theorem:mendelson}
{Let $s \in \N$. Assume that $z_1,\ldots,z_n$ are centered random variables with  variance $1$ that fulfill, for every $i \in [n]$ and $1\leq p \leq \log(en/s)$, 
\begin{equation} \label{eq:subgaussian:momentbound}
    \|z_i\|_{L_p} \leq \lambda \sqrt{p}.
\end{equation} 
Then, there exists an absolute constant $C > 0$ for which the following holds:}
\begin{equation} \label{eq:theorem:mendelson}
\mathbb{E}\left[\sum_{i=1}^s \left(z_i^{*}\right)^2\right]^{1/2}\leq C \lambda \sqrt{ s\log\left(\frac{\e n}{s}\right)},
\end{equation}
where $z_i^{\ast}$ denotes the $i$-th coordinate of the non-increasing rearrangement of the vector $\f{z} = (z_1,\ldots,z_n)$, or, in other words, the $i$-th order statistic of the vector $\f{z} = (z_1,\ldots,z_n)$.
\end{theorem}

{\Cref{theorem:mendelson} improves on Mendelson's \cite[Lemma 6.5]{MendelsonLearning15}, which is an ingredient for the learning theory results of \cite{MendelsonLearning15} and for the sparse recovery guarantees of \cite{DLR18} for measurement matrices fulfilling weak moment assumptions.}

The improvement with respect to \cite[Lemma 6.5]{MendelsonLearning15} is twofold: First, \Cref{theorem:mendelson} \emph{does not require} independence of the $z_i$, and \emph{neither do we require} the $z_i$ to have identical distributions. Both of these are, on the hand, requirements of \cite[Lemma 6.5]{MendelsonLearning15}. Moreover, in order to establish \cref{eq:theorem:mendelson}, \cite[Lemma 6.5]{MendelsonLearning15} requires the bound for the first $O(\log(n))$ moments, whereas \Cref{theorem:mendelson} only requires such a bound for the first $O(\log(\e n/s))$ moments.

While the focus of this manuscript is the generalization of recovery guarantees for $\ell_1$-synthesis, we would like to point the reader interested in the theory of compressive sensing in general to \Cref{D_satisfies_NSP_whp} below for a generalizing of previous results using \Cref{theorem:mendelson}.

{\begin{remark}
It may be of independent interest to note that while the proof of \Cref{theorem:mendelson} uses an assumption about \emph{subgaussian moment} growth for the first few moments in \eqref{eq:subgaussian:momentbound}, it could be adapted to moment bounds of the type $p^{\alpha}$ for some $\alpha > \frac{1}{2}$, if needed.
\end{remark}
}

We present the proof of \Cref{theorem:mendelson} in \Cref{sec:proofs}. We end this section by observing that it is not clear whether our results can be relaxed to allow the entries of the vector $\varphi$ to be arbitrarily dependent.

\section{Random Dictionaries} \label{sec:results:random}
The results we presented in \Cref{sec:results:deterministic} are all based on the dictionary $\f{D}$ satisfying a null space property (NSP). {Proposition} \ref{prop:ChenWangWang} further suggests that the spark of $\f{D}$ plays an important role in the connection between a NSP of the measurement-dictionary product matrix $\f{\Phi} \f{D}$ and recovery guarantees for $\ell_1$-synthesis. However, it is in general NP-hard to verify whether $\f{D}$ is of full spark or whether $\f{D}$ fulfills a null space property \cite{Tillmann13}. 

This problem can typically be avoided by considering \emph{random} matrices, similar to what has been presented in \Cref{sec:results:deterministic} for the measurement matrices. In this way, the desired property can be established at least with high probability. In this section, we therefore consider different setups where \emph{both the dictionary $\f{D}$ and the measurement matrix $\f{\Phi}$} are sampled at random.

In particular, we suppose that the dictionary $\mathbf{D} = \frac{1}{\sqrt{d}}[\mathbf{\psi}_1, \ldots, \mathbf{\psi}_d]^{T} \in \mathbb{R}^{d \times n}$ is now generated at random independently of $\f{\Phi}$ such that each $\mathbf{\psi}_i$ is an independent copy of a random vector $\mathbf{\psi}\in \mathbb{R}^d$. The normalization by $\sqrt{d}$ is not really necessary but it maintains the results in the same line of the previous results. For example, even for a standard Gaussian matrix $\f{G} \in \mathbb{R}^{d\times n}$, the $\ell_2$-norm of the columns are of order $\sqrt{d}$, so the maximum $\rho$ of the column $\ell_2$-norms fulfills $\rho =\Omega(\sqrt{d}) \gg 1$. Without normalization, the scaling $\sqrt{d}$ would appear in the constant $\tau$ of the robust NSP.

\subsection{Properties of Random Dictionaries} \label{sec:prop:rand:dict}
In this section, we establish the robust null space property (robust NSP) of random dictionaries $\f{D}$ whose rows are independent and fulfill weak moment assumptions and a small-ball assumption.

\begin{assumption}[$\f{D}$ with independent rows with $\log(\e n/s)$ well-behaved moments fulfilling {small}-ball assumption] \label{assumption:3}

Let $s \in \N$, $0 < \gamma < 1$ and $S_{\gamma} := \left\{\mathbf{x} \in \R^{n}: \|\mathbf{x}_T\|_2 >\frac{\gamma}{\sqrt{s}}\|\mathbf{x}_{T^c}\|_1 \ \text{for some} \ |T|\le s \right\} \cap \mathbb{S}^{n-1}$, let $\f{D} \in \R^{d \times n}$ be a matrix with i.i.d. rows distributed as $\psi/\sqrt{d} \in \mathbb{R}^n$, where $\psi$ satisfies the conditions
\begin{enumerate}
    \item  $\mathbb{E}[\f{\psi}]=\f{0}$,
    \item  There exists constants $K,c>0$ such that $\inf_{\f{x}\in S_{\gamma}}\mathbb{P}\big(|\langle \f{\psi} , \f{x} \rangle| \ge K\big) \ge c$,
    \item The entries $\xi_i$  of the vector $\psi=(\xi_1,\ldots,\xi_n)$ (not necessarily independent) satisfy $\|\xi_i\|_{L^p} \le \lambda p^{\beta}$ for some $\lambda>0$, $\beta \ge 1/2$, for all $2\le p \lesssim \log \frac{\e n}{s}$. \end{enumerate}
\end{assumption}

The main theorem related to the \Cref{assumption:3} can be stated as follows.

\begin{theorem}
\label{D_satisfies_NSP_whp}
Suppose $\mathbf{D} \in \mathbb{R}^{d\times n}$ fulfills \Cref{assumption:3} {for some $\beta \ge \frac{1}{2}$}. If 
\begin{equation*}
{    d=\Omega\left(\max\left\{s\log\Big(\frac{\e n}{s}\Big), \log\Big(\frac{\e n}{s}\Big)^{\max(2\beta-1,1)}\right\}\right)},
\end{equation*}
then, with probability at least $1-e^{-\Omega(d)}$, $\f{D}$ fulfills the robust NSP of order $s$ with constants $0 < \gamma < 1$ and $\tau > 0$.\footnote{$\tau$ is an absolute constant that only depends on the distributional parameters of $\psi$.}
\end{theorem}

It is worth putting \Cref{D_satisfies_NSP_whp} into context of existing results for standard compressed sensing as introduced in \Cref{sec:sparseintro}. \Cref{D_satisfies_NSP_whp} provides an improvement of \cite[Corollary 8]{DLR18} because it allows dependency among entries of $\psi$ and a bound on fewer moments of $\psi$, i.e., only $\log(\e n/s)$ instead of $\log(n)$. Strictly speaking, our result does not improve \cite[Theorem A]{ML17} because \Cref{D_satisfies_NSP_whp} assumes a slightly stronger condition: The small-ball assumption (2) of \Cref{assumption:3} contains an infimum over the set $S_{\gamma}$ which is smaller than the $\ell_2$-unit ball of Assumption 1, but, on the other hand, $S_{\gamma}$ is taken on a \emph{larger} set than in \cite[Theorem A]{ML17}, where the assumption is made with respect to a set of $s$-sparse vectors. In particular, the authors in \cite{ML17} assume the small-ball assumption on the set $\Sigma_s$ of $s$-sparse vectors in the sphere while our condition is satisfied, for example, if the small-ball assumption holds in the convex hull of $\Sigma_s$ intersected with the $\ell_2$-sphere. On the other hand, \Cref{D_satisfies_NSP_whp}, compared to \cite[Theorem A]{ML17}, requires fewer moments on $\psi$ and decreases the factor $4\beta -1$ to $2\beta-1$ together with a significant improvement in the failure probability.

\begin{remark}
One might wonder how close \Cref{assumption:3} comes to a \emph{minimal} assumption {for establishing the}  null space property of order $s$ as in \Cref{D_satisfies_NSP_whp}. \cite[Theorem C]{ML17} establishes that in fact, a random vector with independent entries whose first $\log(n)/\log(\log(n))$ moments are sub-Gaussian, with a probability of at least $1/2$, does \emph{not} lead to an exact reconstruction property (which is even slightly weaker than our robust NSP). In this sense, our sufficient condition involving the first $\log(\e n/s)$ moments closes parts of this gap. {On the other hand, the typical scenario in compressed sensing is such that $s \ll n$, so that our improvement from $\log(n)$ to $\log(\e n/s)$ might not be numerically relevant in practical scenarios.}
\end{remark}

We recall that {Proposition} \ref{prop:ChenWangWang} suggests that the connection between a null space property $\f{\Phi} \f{D}$ and recovery guarantees of $\ell_1$-synthesis is also related to the spark of $\mathbf{D}$. In order for $\f{D}$ to have full spark, we note that is sufficient to choose $\psi$ to fulfill any joint law such that, for every collection of $d$ entries, the joint law of such collection is absolutely continuous with respect to the Lebesgue measure in $\mathbb{R}^d$ \cite{Blumensath09,ChenWangWang14}.
We note that condition on $\psi$ of \Cref{assumption:3} can be easily obtained with independent entries, and the fact that $\psi$ is drawn from a continuous distribution does not hurt the small-ball assumption (2) of \Cref{assumption:3}. For example, every $\psi$ drawn from a continuous distribution with bounded density automatically satisfies the small-ball assumption \cite{DLR18}.

\subsection{Main Results for Random Dictionaries} \label{sec:mainresults:randomdict}
In this section, we present our main results when the dictionary $\mathbf{D} \in \mathbb{R}^{d\times n}$ is generated at random. Our first model, addressed in \Cref{main_result_first_random_dictionary}, consists in $\mathbf{D}$ generated by $\psi$ as in \Cref{assumption:3} together with the measurement matrix $\f{\Phi} \in \mathbb{R}^{m \times d}$ of the previous chapter. The second model, addressed in \Cref{main_result_second_random_dictionary}, assumes stronger conditions for $\mathbf{D}$ but it will considerably relax the assumptions on the rows $\varphi$ of the measurement matrix $\f{\Phi}$.

We provide our first main result. In a nutshell, it establishes sparse recovery via the $\ell_1$-synthesis method \cref{l1_synthesis_method} with high probability. To the best of our knowledge, it is the first result that considers both dictionary and measurement matrix with heavy tails.
\begin{theorem}
\label{main_result_first_random_dictionary}
Suppose $\f{\Phi} \in \mathbb{R}^{m \times d}$ is a measurement matrix satisfying \Cref{assumption:1} {for some $\alpha \ge \frac{1}{2}$}. Generate $\mathbf{D} = \frac{1}{\sqrt{d}} [\psi_1,\ldots,\mathbf{\psi}_n]^{T} \in \mathbb{R}^{d\times n}$ independently, such that it fulfills \Cref{assumption:3} {for some $\beta \ge \frac{1}{2}$}. 
If the number of rows of $\mathbf{D}$ satisfies
\begin{equation*}
   { d = \Omega \left(\max \left\{s\log(\frac{\e n}{s}), \log(\frac{\e n}{s})^{\max(2\beta-1,1)}\right\}\right)}
\end{equation*}
and the number of rows of $\f{\Phi}$ satisfies
\begin{equation*}
   { m = \Omega \left(\max \left\{s\log(\frac{\e n}{s}), \log(\frac{\e n}{s})^{\max(2\alpha-1,1)}\right\}\right)},
\end{equation*}
then, with probability at least { $1-e^{-\Omega(d)} - e^{-\Omega(m)}$}, $\mathbf{D}$ satisfies {robust} NSP with constants $\gamma$ and some $\widetilde{\tau}$ and furthermore, the $\ell_1$-synthesis method \cref{l1_synthesis_method} provides stable and robust reconstruction of any coefficient vector $\f{x}_0 \in \R^{n}$ and signal $\f{z}_0 \in \R^{d}$ from $y=\f{\Phi} \mathbf{z}_0 +\f{e}$ with $\|\f{e}\|_2 \le \varepsilon$ such that
\begin{equation*}
    \|\hat{\mathbf{x}} - \mathbf{x}_0\|_2 \lesssim \frac{\sigma_s(\f{x}_0)_1}{\sqrt{s}}  + \widetilde{\tau} \epsilon 
\end{equation*}
\begin{equation*}
    \|\hat{\mathbf{z}} - \mathbf{z}_0\|_2 \lesssim \|\mathbf{D}\|\left(\frac{\sigma_s(\f{x}_0)_1}{\sqrt{s}}  + \widetilde{\tau} \epsilon \right).
\end{equation*}
In particular, it holds that $\mathbb{E}\|\hat{\mathbf{z}} - \mathbf{z}_0\|_2 \lesssim \mathbb{E}[\|\mathbf{D}\|](\sigma_{s}(\mathbf{x}_0)/\sqrt{s} + \widetilde{\tau} \varepsilon)$.
\end{theorem}
The proof follows easily from the tools developed to prove the results presented in \Cref{sec:results:deterministic} and \Cref{sec:prop:rand:dict} and is detailed in \Cref{sec:proofs:random:main}.

Next, we consider the same problem as before, but with different assumptions on our dictionary $\mathbf{D}\in \mathbb{R}^{d \times n}$ and the random vector $\varphi \in \mathbb{R}^{d}$ that determines the distribution of the measurement matrix $\f{\Phi}$. This result explores the trade-off between the assumptions on the dictionary $\mathbf{D}$ and the measurement matrix $\f{\Phi}$. Namely, we are now interested in random dictionaries $\mathbf{D}$ with independent entries drawn from a distribution with a few sub-Gaussian moments, and independently sampled measurements $\varphi$ whose distribution possesses only a finite second moment.
\begin{assumption}[$\f{\Phi}$ with i.i.d. entries with finite variance] \label{assumption:4}
Let $\f{\Phi} \in \R^{m \times d}$ be a matrix with i.i.d. rows distributed as $\varphi/\sqrt{m}$, where $\varphi$ satisfies
\begin{enumerate}
    \item  $\mathbb{E}[\f{\varphi}]=\f{0}$,
    \item  There exists constants $K,c>0$ such that $\inf_{\f{x}\in \mathbb{S}^{d-1}}\mathbb{P}(|\langle \f{\varphi} , \f{x} \rangle| \ge K) \ge c$,
    \item The entries $\xi_i$ of $\varphi$ have bounded variance, i.e., there exists $\lambda >0$ that $\|\xi_i\|_{L_2} \leq \lambda$ for all $i \in [d]$.
\end{enumerate}
\end{assumption}
Notice that the assumptions on $\f{\Phi}$ are very mild. Furthermore, a random dictionary $\f{D}$ with i.i.d. entries with $\log(n)$ sub-Gaussian moments satisfies the assumptions of \Cref{D_satisfies_NSP_whp}, for $d$ large enough---the construction corresponds to the case in which the entries of $\psi$ are i.i.d. random variables with moment growth parameter $\beta = 1/2$. In particular, we can show the following.

\begin{theorem}
\label{main_result_second_random_dictionary}
Suppose $\f{\Phi} \in \mathbb{R}^{m \times d}$ is a random matrix satisfying \Cref{assumption:4}, and suppose that $\f{D} = (\mathbf{D})_{ij} = (d^{-1/2}\mathbf{d})_{ij}$ is a dictionary with i.i.d centered entries, such that the $\mathbf{d}_{ij}$ are drawn independently from $\f{\Phi}$, and their distribution's first $\log(n)$ moments are sub-Gaussian, i.e, $\|\mathbf{d}_{ij}\|_{L^p}= O(\sqrt{p})$ for $p\le \log(n)$, for all $i \in [d]$ and $j \in [n]$.

Then, the same conclusions as in \Cref{main_result_first_random_dictionary} hold.
\end{theorem}

We highlight that the results also holds for independent but non identical entries of $\mathbf{D}$, as can be seen in the proof in \Cref{sec:proofs:random:main}. Note that we require in \Cref{main_result_second_random_dictionary} not only $O(\log(\e n/s))$, but $O(\log(n))$ sub-Gaussian moments, which deviates from previous results in this paper. 

\section{Proofs} \label{sec:proofs}
{In this section, we present the proofs of the statements of \Cref{sec:results:deterministic,sec:results:random}.}

\subsection{Proof of Theorem \ref{theorem:mendelson}}

As a preparation for the proof of Theorem \ref{theorem:mendelson}, we present a simple consequence of Markov's inequality.

\begin{lemma}{{\cite[Lemma 3.7]{mendelson2018improved}}}
\label{Lemma3.7_Mendelson}
Assume that a random variable $X$ satisfies $\|X\|_{L^p} \le A$. Then, for each $k>1$,
\begin{equation*}
    \mathbb{P}(|X|\ge k A) \le k^{-p}.
\end{equation*}
\end{lemma}

\begin{proof}[{Proof of \Cref{theorem:mendelson}}]
{Fix $p = \log(\e n/s)$}. Defining $Z_i := z_i \mathds{1}_{|z_i| \le \e \lambda \sqrt{p}}$ for each $i \in [n]$, we obtain that 
\begin{equation} \label{eq:lemma51:proof:1}
     \mathbb{E}\left[\sum_{i=1}^s (Z_i^{*})^2\right] \leq  s \e^2 \lambda^2 p
\end{equation}
if $Z_i^*$ is the $i$-th largest coordinate of the vector $\f{Z}= (Z_1,\ldots,Z_n)$. Now we have to deal with the random variables $Y_i = z_i \mathds{1}_{|z_i| > \e \lambda \sqrt{p}}$. We estimate that
\begin{equation} \label{eq:lemma51:proof:2}
    \mathbb{E}\left[\sum_{i=1}^s (Y_i^{*})^2 \right] \le \mathbb{E} \sum_{i=1}^n Y_i^2 = \sum_{i=1}^n \mathbb{E} Y_i^2.
\end{equation}
Observe that by the H{\"o}lder inequality,
\begin{equation*}
    \mathbb{E}[Y_i ^2] \le \mathbb{E}[|z_i|^{2p}]^{1/p}\mathbb{P}(|z_i|\ge \e \lambda \sqrt{p})^{1-1/p}.
\end{equation*}
{Furthermore, using} the moment assumption {for $p = \log(en/s)$,} we obtain
\begin{equation*}
    \mathbb{E}[|z_i|^{2p}]^{1/p} \leq \lambda^2  (\sqrt{2 p})^2 {\leq 2\lambda^2 \log(\e n/s)}.
\end{equation*}
Applying \Cref{Lemma3.7_Mendelson} for $A = \lambda \sqrt{p}$ and $k = \e $, it further follows that
 \begin{equation*}
     \mathbb{P}(|z_i|\ge e \lambda \sqrt{p})^{1-1/p} \le (\e^{-p})^{1-1/p} = \e^{-p+1}= \frac{\e s}{\e n} = \frac{s}{n}
 \end{equation*}
 and therefore
\begin{equation*}
\sum_{i=1}^n \mathbb{E} Y_i^2 \leq  2 n \lambda^2 \log(\e n/s) \frac{s}{n} = 2 \lambda^2 s \log(\e n/s).
\end{equation*}
To get our desired bound, we {use that $z_i = Z_i + Y_i$, \cref{eq:lemma51:proof:1} and \cref{eq:lemma51:proof:2} to estimate }
\begin{equation*}
    \mathbb{E}\left[\sum_{i=1}^s({z_i^{*}})^2\right] \le \mathbb{E}\left[\sum_{i=1}^s(Z_i^{*})^2\right] + \mathbb{E}\left[\sum_{i=1}^s(Y_i^{*})^2\right] \le {s \e^2 \lambda^2 p + 2 \lambda^2 s \log\left(\frac{\e n}{s}\right) = (2+\e^2) \lambda^2 s \log\left(\frac{\e n}{s}\right),} 
\end{equation*}
and take the square root on both sides of the inequality. Using Jensen's inequality, we obtain the desired result {with constant $C = \sqrt{2 + \e^2}$.}
\end{proof}

\subsection{Application of Theorem \ref{theorem:mendelson}}
In this section we state and prove an important proposition to our analysis in which Theorem \ref{theorem:mendelson} plays a key role.

{\begin{proposition} \label{Proposition_Main_Result}
Assume that $\varphi_1,\ldots,\varphi_m$ are independent copies of a random vector $\varphi$ fulfilling condition (3) of either \Cref{assumption:1} or \Cref{assumption:2} for some {$\alpha \ge 
\frac{1}{2}$}. Let $\mathbf{D} = [\mathbf{d}_1,\ldots,\mathbf{d}_n] \in \mathbb{R}^{d\times n}$ be a dictionary whose columns have bounded $\ell_2$-norms, i.e., for all $i\in [n]$, $0<\|\mathbf{d}_i\|_2 \le \rho$. Define $V:= m^{-1/2}\sum_{i=1}^m \varepsilon_i \varphi_i$ where the $\epsilon_i$ are independent Rademacher random variables that are independent of the $\{\varphi_i\}_{i=1}^m$.
If also $m\gtrsim \log(\e n/s)^{\max(2\alpha-1,1)}$, then
\begin{equation*}
    \mathbb{E}\left[\sum_{i=1}^s \left((\mathbf{D}^{T}V)_i^{*}\right)^2\right]^{1/2} \lesssim  \rho \sqrt{s\log\Big(\frac{\e n}{s}\Big)},
\end{equation*}
where $(\mathbf{D}^{T}V)_i^{\ast}$ denotes the $i$-th coordinate of the non-increasing rearrangement of the vector $\mathbf{D}^{T}V$.
\end{proposition}}

As a tool to prove {Proposition} \ref{Proposition_Main_Result} in the case of \Cref{assumption:1}, we establish a Khintchine type inequality for $V$.

\begin{lemma}
\label{Easy_Khintchine}
Let $V:= m^{-1/2}\sum_{i=1}^m \varepsilon_i \varphi_i$ where the {$\varphi_i \in \mathbb{R}^d$} are independent and satisfy the conditions of \Cref{assumption:1} {for some $\alpha \ge \frac{1}{2}$} and $\{\varepsilon_i\}_{i=1}^m$ are independent Rademacher random variables. Then, for every $\f{a} \in \mathbb{S}^{d-1}$, if $p\lesssim \log (\frac{\e n}{s})$ and $m\gtrsim \log(\e n/s)^{\max(2\alpha-1,1)}$, where $\alpha$ is the moment growth parameter of \Cref{D_satisfies_NSP_whp}, it holds that
\begin{equation*}
    \|\langle \f{a},V \rangle \|_{L^p} \lesssim \sqrt{p}. 
\end{equation*}
\end{lemma}

\begin{proof}[{Proof of \Cref{Easy_Khintchine}}]
By definition of $V$ we see that
\begin{equation*}
    \|\langle \f{a},V \rangle \|_{L^p} = \left\|\frac{1}{\sqrt{m}}\sum_{i=1}^m\varepsilon_i\langle \f{a}, \varphi_i \rangle \right\|_{L^p}.
\end{equation*}
All random variables $\left\{\varepsilon_i \langle \f{a},\varphi_i \rangle\right\}_{i=1}^m$ are i.i.d. copies of a random variable with the first $p$ moments of order $p^{\alpha}$ due to condition (3) of \Cref{assumption:1}. Thus, by \cite[Lemma 2.8]{ML17}, there exists a constant $c(\alpha)$ that depends only on $\alpha$ such that
\begin{equation*}
    \left\|\frac{1}{\sqrt{m}}\sum_{i=1}^m\varepsilon_i\langle \f{a}, \varphi_i \rangle \right\|_{L^p} \leq c(\alpha)\lambda \sqrt{p}.
\end{equation*}
\end{proof}

To show {Proposition} \ref{Proposition_Main_Result}, we need a statement analogue to \Cref{Easy_Khintchine} that holds for {$\varphi_i$} as in \Cref{assumption:2}. {The} challenge is that a moment bound on the marginals $\langle \f{a},\varphi_i \rangle$ is not directly given by \Cref{assumption:2}. We use a moment comparison argument to overcome this issue and establish a Khintchine inequality under weak moment assumptions (weaker than the standard subgaussian assumption). Similar arguments have already appeared in the literature, see \cite[Chapter 3]{kwapien1992random}.

\begin{lemma}[Khintchine inequality under weak moment assumption]
\label{weak_khintchine}
Suppose $X=(x_1,\ldots,x_d) \in \mathbb{R}^d$ with independent mean zero $x_i$ and also assume, for all $i\in [d]$, $\|x_i\|_{L^p} = O(\sqrt{p})$ for $1\le p\le k$. Then, for all vectors $\f{a}\in \mathbb{R}^d$, we have
\begin{equation*}
    \|\langle X,\f{a}\rangle\|_{L^p} = O(\sqrt{p} \|\f{a}\|_2), \quad p\in [1,k].
\end{equation*}
\end{lemma}
\begin{proof}[{Proof of \Cref{weak_khintchine}}]
By a standard symmetrization argument 
\[
\|\langle X,\f{a}\rangle\|_{L^p} = \Big\|\sum_{i=1}^d \f{a}_i x_i \Big\|_{L^p} \le 2\Big\|\sum_{i=1}^{d} \varepsilon_i \f{a}_i x_i\Big\|_{L^p},
\]
where $\varepsilon_i$ are independent standard Rademacher random variables. So we may assume that our vector $X$ has independent symmetric coordinates, i.e, for all vectors $\f{a}$, the random variables $\langle X,\f{a} \rangle$ and $-\langle X,\f{a} \rangle$ have the same distribution. Without loss of generality assume $p$ to be an even integer. Then, by expanding in monomials,
\begin{equation*}
    \mathbb{E}\left[\bigg(\sum_{i=1}^d \f{a}_i x_i\bigg)^p\right] = \sum_{i=1}^{\min(p,d)} \sum_{i_1<i_2<\ldots<i_t}\sum_{d_1+\ldots +d_t =p, d_i\ge 1} \binom{p}{d_1,\ldots,d_t}\left(\prod_{j=1}^t \f{a}_{i_j}^{d_j} \mathbb{E} \left[x_{i_j}^{d_j}\right]
    \right)
\end{equation*}
Any odd moment vanishes due to the symmetry. The even moments $\mathbb{E} [x_{i_j}^{d_j}]$ (until "k") are subgaussian by hypothesis, so it follows that the $p$-norm of our inner product $\langle X,\f{a}\rangle$ is dominated (up to an absolute constant) by the $p$-norm of $\langle g , \f{a}\rangle$, where $g$ is a standard Gaussian vector. Precisely, we obtain that
\begin{equation*}
    \mathbb{E}\left[\bigg(\sum_{i=1}^d \f{a}_i x_i\bigg)^p\right] \lesssim \sum_{i=1}^{\min(p,d)} \sum_{i_1<i_2<\ldots<i_t}\sum_{d_1+\ldots +d_t =p, d_i\ge 1} \binom{p}{d_1,\ldots,d_t}\left(\prod_{j=1}^t \f{a}_{i_j}^{d_j} \mathbb{E} \left[g_{i_j}^{d_j}\right]
    \right) = \mathbb{E}[\langle g, \f{a}\rangle^p]
\end{equation*}
{Note} that the random variable $\langle g, \f{a}\rangle$ follows a Gaussian distribution with zero mean and variance $\|\f{a}\|_2^2$. The proof now follows by standard estimates for the $p$-norms of Gaussian distributions. 
\end{proof}

\medskip

Using \Cref{weak_khintchine}, we now establish a lemma that is analogous to \Cref{Easy_Khintchine}.
\begin{lemma} \label{lemma:assumption2:khintchine}
Let $V:= m^{-1/2}\sum_{i=1}^m \varepsilon_i \varphi_i$, where the {$\varphi_i \in \mathbb{R}^d$} are independent satisfying the conditions of \Cref{assumption:2} {for some $\alpha \ge \frac 12$} and $\{\varepsilon_i\}_{i=1}^m$ are independent Rademacher random variables. Suppose that $m\gtrsim \log(\e n/s)^{\max(2\alpha-1,1)}$. Then, for every $\f{a} \in \mathbb{S}^{d-1}$ and $p\lesssim \log (\frac{\e n}{s})$, 
\begin{equation*}
    \|\langle \f{a},V \rangle \|_{L^p} \lesssim \sqrt{p}. 
\end{equation*}
\end{lemma}

\begin{proof}[{Proof of \Cref{lemma:assumption2:khintchine}}]
Observe that the entries {$V_j := m^{-1/2}\sum_{i=1}^m \varepsilon_i (\varphi_i[j])$} of $V$ are a normalized sum of i.i.d random variables with moments {bounds} of order $p^{\alpha}$ due to \Cref{assumption:2}. By {\cite[Lemma 2.8]{ML17}}, we know that $\|V_j\|_{L^p} = O(\sqrt{p})$, for $p=O(\log(\e n/s))$. Moreover, by independence of the entries of the $\varphi_i$, the entries $V_j$ are also independent, therefore Lemma \ref{weak_khintchine} can be applied to finish the proof.
\end{proof}
\medskip

We now go ahead and combine the previous results to show {Proposition} \ref{Proposition_Main_Result}.
\medskip

\begin{proof}[{Proof of {Proposition} \ref{Proposition_Main_Result}}]
Let $V:= m^{-1/2}\sum_{i=1}^m \varepsilon_i\varphi_i$, {which} is defined with independent $\varphi_i$ that are centered random vectors fulfilling the last condition of \Cref{assumption:1} or of \Cref{assumption:2} and the $\epsilon_i$ are independent Rademacher random variables that are independent from the vectors $\{\varphi_i\}_{i=1}^m$. Let $\mathbf{D} = [\mathbf{d}_1,\ldots,\mathbf{d}_n] \in \mathbb{R}^{d\times n}$ be a dictionary such that $0<\|\mathbf{d}_i\|_2 \le \rho$ for all $i\in [n]$. First, we note that using the assumptions and \Cref{Easy_Khintchine} or \Cref{lemma:assumption2:khintchine}, respectively, we see that

\[
\|  \langle \f{d}_i , V \rangle \|_{L^2} =  \|\f{d}_i\|_2  \left\| \langle \f{d}_i / \|\f{d}_i\|_2, V \rangle \right\|_{L^2} \lesssim  \rho.
\]
{For all $i\in [n]$ such that $\|\langle \f{d}_i , V \rangle\|_{L_2} >0$, we define $z_i := \left(\f{D}^T V\right)_i/\|\langle \f{d}_i , V \rangle\|_{L_2} = \langle \f{d}_i , V \rangle / \|\langle \f{d}_i , V \rangle\|_{L_2}$}.  We can again use \Cref{Easy_Khintchine} or \Cref{lemma:assumption2:khintchine}, respectively, to see that $z_i$ fulfills the assumptions of \Cref{theorem:mendelson} in each case, in particular $\|z_i\|_{L_p} \leq C \rho \sqrt{p}$ for all $i \in [n]$ for some constant $C > 0$. This then implies
\begin{equation*}
  { \mathbb{E}\left[\sum_{i=1}^s \left((\mathbf{D}^{T}V)_i^{*}\right)^2\right]^{1/2} = \mathbb{E}\left[\sum_{i=1}^s \left(z_i^{*} \|\langle \f{d}_i , V \rangle\|_{L_2}\right)^2\right]^{1/2} \lesssim \rho \sqrt{ s\log(\frac{\e n}{s})}},
\end{equation*}
which is the desired bound.
\end{proof}

\subsection{Proofs of \Cref{main_result_deterinistic_dictionary,main_corollary_deterministic_dictionary}} \label{sec:proofs:deterministic}
In this section, we detail the proofs of \Cref{main_result_deterinistic_dictionary,main_corollary_deterministic_dictionary}.

The main idea of our proofs is based on the so-called \emph{small-ball method} that was first developed by Koltchinskii and Mendelson \cite{koltchinskii2015bounding,MendelsonLearning15,ML17}. As a preparation, we consider the following definitions.

\begin{definition}{(Marginal tail function)}
For a random vector $\psi \in \mathbb{R}^d$, a subset $S\subset \mathbb{R}^d$ and a fixed $A>0$, define the marginal tail function $Q_{A}(S;\f{\psi})$ by
\begin{equation*}
    Q_{A}(S,\psi) := \inf_{\f{x} \in S} \mathbb{P}(|\langle \f{x},\psi\rangle| \ge A).
\end{equation*}
\end{definition}

\begin{definition}{(Mean empirical width {\cite{Tropp15}, Rademacher complexity)}}
For a random vector $\f{\psi} \in \mathbb{R}^d$ and a subset $S\subset \mathbb{R}^d$, the mean empirical width of $S$ is
\begin{equation*}
    W_m(S,\f{\psi}) := \mathbb{E}\sup_{\f{x}\in S} \Big\langle \f{x}, \frac{1}{m} \sum_{i=1}^m \varepsilon_i \f{\psi}_i \Big\rangle,
\end{equation*}
where $\f{\psi}_1,\ldots,\f{\psi}_m$ are i.i.d. copies of $\f{\psi}$ and $\varepsilon_1,\ldots,\varepsilon_m$ are indepdendent Rademacher random variables that are independent of the $\{\psi_i\}_{i=1}^m$.
\end{definition}
The following proposition is at the core of our proof.
\begin{proposition}[{\cite[Theorem 1.5]{koltchinskii2015bounding},\cite[Proposition 5.1]{Tropp15}}]
\label{small_ball_method}
Fix a set $S \subset \mathbb{R}^n$. Let $\psi \in \mathbb{R}^n$ be a random vector and let $\Psi \in \mathbb{R}^{m\times n}$ be a random matrix whose rows are i.i.d copies of $\psi$. Then, for any $t>0$ and $A>0$,
\begin{equation*}
    \inf_{\f{x} \in S} \|\Psi \f{x}\|_2 \ge A \sqrt{m} Q_{2A}(S;\psi) -2 \sqrt{m} W_m(S;\psi) - A t,
\end{equation*}
with probability at least $1-e^{-t^2/2}$.
\end{proposition}

\medskip

\begin{proof}[{Proof of \Cref{main_result_deterinistic_dictionary}}]
Since $\f{D} = [\mathbf{d}_1,\ldots,\mathbf{d}_n] \in \R^{d \times n}$ satisfies the robust NSP of order $s$ with constants $\gamma$ and $\tau$, we know that
    \begin{equation} \label{eq:NSP:eq:proofthm1}
    \inf_{\f{x} \in S_{\gamma}} \|\mathbf{D}\mathbf{x}\|_2 \ge \frac{1}{\tau}
	\end{equation}
for the set $S_{\gamma} := \left\{\mathbf{x} \in \R^{n}: \|\mathbf{x}_T\|_2 >  \frac{\gamma}{\sqrt{s}}\|\mathbf{x}_{T^c}\|_1 \ \text{for some} \ |T|\le s \right\} \cap \mathbb{S}^{n-1}$.

To show that $\f{\Phi}\f{D}$ fulfills the robust NSP of order $s$ with constants $\gamma$ and $\widetilde{\tau}$, we need to show that
\[
{\inf_{x \in S_{\gamma}}{\|\f{\Phi}\mathbf{D}x\|_2} \ge \frac {1}{\widetilde{\tau}}.}
\]
To do this, we will use {Proposition \ref{small_ball_method}} for $\Psi = \f{\Phi}\f{D}$ and $\f{\psi} = \f{D}^T \f{\varphi}/\sqrt{m}$ and $S = S_{\gamma}$.

We first estimate the mean empirical width $W_m(S_{\gamma},\mathbf{D}^{T}\varphi/\sqrt{m})$. We consider the sets $\Sigma_{s} := \{ \mathbf{x}: \|\f{x}\|_{0}\le s, \ \|\mathbf{x}\|_{2}=1\}$, $T_{\gamma,s} = \big\{ \mathbf{x}: \|\mathbf{x}_{T}\|_2\ge \frac{\gamma}{\sqrt{s}}\|\mathbf{x}_{T^c}\|_1 \text{ for some} \ |T|\le s\big\}$ and $D_s$ the convex hull of the set $\Sigma_2$.{ By {\cite[Lemma 2]{DLR18}} we have $T_{\gamma,s}\cap \mathbb{B}_2^n \subset (2+\gamma^{-1})D_s$}. Therefore, since $S_{\gamma} \subset T_{\gamma,s}\cap \mathbb{B}_{2}^n$, by H{\"o}lder's inequality we get
\begin{equation*}
    W_m(S_{\gamma},\mathbf{D}^{T}\varphi/\sqrt{m}) \le W_m(T_{\gamma,s} \cap \mathbb{B}_2^n, \mathbf{D}^{T}\varphi/\sqrt{m}) \le  (2+\gamma^{-1})W_m( D_s,\mathbf{D}^{T}\varphi/\sqrt{m}).
\end{equation*}
Let $V:= m^{-1/2}\sum_{i=1}^m \varepsilon_i\varphi_i$. Since $\sum_{i=1}^m \varepsilon_i \f{D}^T \varphi_i/\sqrt{m} = \f{D}^T V$,
we see that
\begin{equation}
\label{main_problem_paper}
    W_m( D_s,\mathbf{D}^{T}\varphi/\sqrt{m}) = m^{-1}\mathbb{E}\sup_{\mathbf{x} \in D_s} \langle \mathbf{x}, \f{D}^T V \rangle  = m^{-1}\mathbb{E}\sup_{\mathbf{x} \in \Sigma_s} \langle \mathbf{x},\mathbf{D}^{T} V \rangle = m^{-1}\mathbb{E}\left[\sum_{i=1}^s \left((\mathbf{D}^{T}V)_i^{*}\right)^2\right]^{1/2} .
\end{equation}
In the second equality, we used the fact that the supremum of the linear form over the convex hull $D_s$ and over $\Sigma_s$ coincides. Under the assumptions of \Cref{main_result_deterinistic_dictionary}, it holds that $m\gtrsim \log(\e n/s)^{\max(2\alpha-1,1)}$ and therefore {Proposition \ref{Proposition_Main_Result}} implies that
\[
W_m(S_{\gamma},\mathbf{D}^{T}\varphi/\sqrt{m}) \leq  \frac{(2+\gamma^{-1})}{m}  \mathbb{E}\left[\sum_{i=1}^s \left((\mathbf{D}^{T}V)_i^{*}\right)^2\right]^{1/2} \lesssim  \frac{(2+\gamma^{-1})\rho}{m} \sqrt{s\log\Big(\frac{\e n}{s}\Big)},
\]
where $\rho = \Theta(1)$ is the uniform upper bound on the $\ell_2$-norms of the columns of $\f{D}$.

Furthermore, for $A = \frac{K}{2 \tau\sqrt{m}} $, where $K$ and $c$ are the constants of condition (2) of \Cref{assumption:1}, we lower bound $Q_{2A}(S_{\gamma};\f{D}^T \varphi)$
\[
\begin{split}
Q_{2A}(S_{\gamma};\f{D}^T \varphi/\sqrt{m}) &= \inf_{\f{x} \in S_{\gamma}} \mathbb{P}(|\langle \f{x},\f{D}^T\varphi/\sqrt{m}\rangle| \ge 2 A) = \inf_{\f{x} \in S_{\gamma}} \mathbb{P}(|\langle \f{D}\f{x},\varphi/\sqrt{m}\rangle| \ge 2 A) \\
&= \inf_{\f{x} \in S_{\gamma}} \mathbb{P}\left(\Big|\Big\langle \frac{\f{D}\f{x}}{\|\f{D}\f{x}\|_2},\varphi\Big\rangle\Big| \ge \frac{2 A \sqrt{m}}{\|\f{D}\f{x}\|_2}\right) \geq \inf_{\f{z} \in \mathbb{S}^{n-1}}  \mathbb{P}\left(|\langle \f{z},\varphi\rangle| \ge 2 A \tau \sqrt{m} \right) \geq c,
\end{split}
\]
Putting this together with the above estimate for the mean empirical width, we obtain by {Proposition \ref{small_ball_method}} that, with probability $1-e^{-t^2}$, there exits an absolute constant $C>0$ such that
\begin{equation*}
    \inf_{x \in S_{\gamma}}{\|\Phi\mathbf{D}x\|_2} \ge \frac{K c}{2\tau} - 2\sqrt{m} W_m(S_{\gamma},\mathbf{D}^{T}\varphi) - \frac{K}{2\tau \sqrt{m}} t \ge \frac{K c}{2\tau}-\frac{C(2+\gamma^{-1})\rho}{\sqrt{m}} \sqrt{s\log\Big(\frac{\e n}{s}\Big)} - \frac{K}{2\tau \sqrt{m}} t.
\end{equation*} 

{Now we choose $t = \frac{c\sqrt{m}}{4}$ and $m = \frac{16\tau^2(C(2+\gamma^{-1})\rho)^2 s\log(\frac{\e n}{s})}{(Kc)^2}$ to obtain that}

\begin{equation*}
    \inf_{x \in S_{\gamma}}{\|\Phi\mathbf{D}x\|_2} \ge \frac{3Kc}{16} \left(\frac{1}{\tau}\right)  .
\end{equation*}

with probability at least $1-e^{-mc^2/8}$. This concludes the proof as it shows that under the conditions stated by \Cref{main_result_deterinistic_dictionary}, with probability at least $1-e^{-\Omega(m)}$, $\f{\Phi} \f{D}$ fulfills the robust NSP of order $s$ with constants $\gamma$ and $\widetilde{\tau} := \frac{16}{3Kc}\tau = \Theta (\tau)$.

Finally, on this event, the two inequalities bounding the distances $\|\f{\hat{x}} - \f{x}_0\|_2$ and $\|\f{\hat{z}} - \f{z}_0\|_2$ between the output of $\ell_1$-synthesis \cref{l1_synthesis_method} $\f{\hat{x}}$ and $\f{\hat{z}}$ and follow from {Proposition \ref{prop:recguarantee}} and the fact that $\|\f{\hat{z}} - \f{z}_0\|_2$ = $\|\f{D}(\f{\hat{x}} - \f{x}_0)\|_2 \leq \|\f{D}\| \|\f{\hat{x}} - \f{x}_0\|_2$.
\end{proof}

\medskip

We continue with the proof of \Cref{main_corollary_deterministic_dictionary}.

\medskip

\begin{proof}[{Proof of \Cref{main_corollary_deterministic_dictionary}}]
To show the statement of \Cref{main_corollary_deterministic_dictionary}, we proceed analogously to the proof of \Cref{main_result_deterinistic_dictionary}. \Cref{main_result_deterinistic_dictionary} can be used to upper bound the mean empirical width $W_m(S_{\gamma},\mathbf{D}^{T}\varphi/\sqrt{m})$, but to lower bound the marginal {tail} function, we proceed slightly differently. 

Similarly as in the proof of \Cref{weak_khintchine} for a unit norm $\f{a} \in \mathbb{S}^{d-1}$, we {have}
\begin{equation*}
    \mathbb{E}\left[|\langle \varphi, \f{a} \rangle |^4\right] = \sum_{i=1}^{d} \sum_{i_1<i_2<\ldots<i_t}\sum_{d_1+\ldots d_t =4, d_i\ge 1} \binom{4}{d_1,\ldots,d_t}\left(\prod_{j=1}^t \f{a}_{i_j}^{d_j} \mathbb{E}[ \xi_{i_j}^{d_j}]\right)
\end{equation*}
By a standard symmetrization argument, we may assume that the $\xi_i$ are symmetric random variables. Using the symmetry and {(1) of \Cref{assumption:2}}, we note that the $\mathbb{E}[ \xi_{i_j}^{d_j}]$ are zero for all odd $d_j$. The only possibilities for non-zero terms are such that either
\begin{enumerate}
	\item There exists $i_j \in \N$ such that $i_j = 4$ and $\prod_{j=1}^t \f{a}_{i_j}^{d_j} \mathbb{E}[ \xi_{i_j}^{d_j}] = \f{a}_{i_j}^4 \mathbb{E}[ \xi_{i_j}^4]$. By {moment} condition (3) of \Cref{assumption:2}, $\f{a}_{i_j}^4 \mathbb{E}[ \xi_{i_j}^4] \leq \f{a}_{i_j}^4 \lambda^4 4^{4 \alpha}$.
	\item There exist $i_j \neq i_k \in \N$ such that $d_{i_j}=d_{i_k} = 2$ and $\prod_{j=1}^t \f{a}_{i_j}^{d_j} \mathbb{E}[ \xi_{i_j}^{d_j}] = \f{a}_{i_j}^2 \f{a}_{i_k}^2 \mathbb{E}[ \xi_{i_j}^2 \xi_{i_k}^2] = \f{a}_{i_j}^2 \f{a}_{i_k}^2 \mathbb{E}[ \xi_{i_j}^2]  \mathbb{E}[\xi_{i_k}^2] \leq \f{a}_{i_j}^2 \f{a}_{i_k}^2 \lambda^2 2^{4 \alpha} \lambda^2 2^{4 \alpha} = \f{a}_{i_j}^2 \f{a}_{i_k}^2 \lambda^4 4^{4 \alpha}$, using independence of the $\xi_i$.
\end{enumerate}
{Therefore,
\[
\mathbb{E}\left[|\langle \varphi, \f{a} \rangle |^4\right]  \leq \lambda^4 4^{4 \alpha} \left(\|\f{a}\|_4^4 + \|\f{a}\|_2^4\right) \le 2\lambda^4 4^{4 \alpha}
\]
and thus
\[
\|\langle \varphi, \f{a} \rangle |\|_{L_4} \leq 2^{1/4}\lambda 4^{\alpha}. 
\]
Furthermore, by independence of the $\{\xi_i\}_{i=1}^m$ it follows that $\|\langle \varphi, \f{a} \rangle |\|_{L_2} = \|\f{a}\|_2 = 1$, {using (2) of \Cref{assumption:2}}. Finally, using the Paley-Zygmund inequality (cf., e.g., \cite[Lemma 7.16]{foucart2013invitation}), we conclude that
\[
\mathbb{P}(|\langle \f{\varphi} , \f{a} \rangle| \ge 1/2) \geq \frac{\big(\mathbb{E}\left[|\langle \varphi, \f{a} \rangle |^2\right] - 1/2\big)^2}{\mathbb{E}\left[|\langle \varphi, \f{a} \rangle |^4\right]}	 \geq \frac{1}{4} \frac{1}{2^{1/4}\lambda 4^{\alpha}},
\]}
which shows that $\varphi$ actually fulfills condition (2) of \Cref{assumption:1} for $K = 1/2$ and {$c = 2^{-1/4}\lambda^{-1} 4^{-\alpha-1}$}. Therefore, the statement of \Cref{main_corollary_deterministic_dictionary} follows from \Cref{main_result_deterinistic_dictionary}.
\end{proof}

\subsection{Proof of  \Cref{D_satisfies_NSP_whp}} \label{sec:proof:thm3}

Here, we show the results of \Cref{sec:prop:rand:dict}. {We start with a proposition that plays a similar role as Proposition \ref{Proposition_Main_Result}.

\begin{proposition}
\label{Proposition_Main_Result_basis_case}
Suppose $\mathbf{D} \in \mathbb{R}^{d\times n}$ fulfills \Cref{assumption:3} {for some $\beta \ge \frac{1}{2}$}, let $V_{\psi}:= \frac{1}{\sqrt{d}} \sum_{i=1}^d \varepsilon_i \psi_i$ with the $\psi_i \in \mathbb{R}^n$ as in \Cref{assumption:3}, and where the $\epsilon_i$ are independent Rademacher random variables that are independent of the $\{\psi_i\}_{i=1}^m$. If $d\gtrsim \log(\frac{\e n}{s})^{\max(2\beta-1,1)}$, then
\begin{equation*}
    \mathbb{E}\Big[\sum_{i=1}^s \big((V_{\psi})_i^{*}\big)^2\Big]^{1/2} \lesssim \sqrt{s\log\Big(\frac{\e n}{s}\Big)}.
\end{equation*}
\end{proposition}}

\medskip

\begin{proof}[{Proof of {Proposition} \ref{Proposition_Main_Result_basis_case}}]
	The proof is very similar to the proof of {Proposition \ref{Proposition_Main_Result}}. In particular, for $V_{\psi}= d^{-1/2} \sum_{j=1}^d \varepsilon_j \psi_j$, {we can define, for each $i \in [n]$ such that $\|\langle \f{e}_i , V_{\psi}\rangle\|_2 >0$, $z_i:= (V_{\psi})_i / \|(V_{\psi})_i\|_2 = \langle \f{e}_i , V_{\psi}\rangle / \|\langle \f{e}_i , V_{\psi}\rangle\|_2$} , we apply \Cref{lemma:assumption2:khintchine} as (3) of \Cref{assumption:3} coincides with (3) of \Cref{assumption:2}. It follows that $\|z_i\|_{L_p} \lesssim \sqrt{p}$ for all $i \in [n]$ if $d\gtrsim \log(\e n/s)^{\max(2\beta-1,1)}$. A direct application of \Cref{theorem:mendelson} shows the result.
\end{proof}

 \medskip
 
\begin{proof}[{Proof of \Cref{D_satisfies_NSP_whp}}]
This statement can be shown analogously to \Cref{main_result_deterinistic_dictionary} by using {Proposition \ref{small_ball_method}} as follows: We choose $\Psi = \f{D}$, whose rows are distributed as $\psi/\sqrt{d}$, and $S = S_{\gamma}$. The mean empirical width can be bounded by {Proposition} \ref{Proposition_Main_Result_basis_case}, we bound the marginal tail function such that
\[
Q_{2A}(S_{\gamma};\psi/\sqrt{d}) = \inf_{\f{x} \in S_{\gamma}} \mathbb{P}(|\langle \f{x},\psi\rangle| \ge 2 A\sqrt{d}) = \inf_{\f{x} \in S_{\gamma}} \mathbb{P}(|\langle \f{x},\psi\rangle| \ge 2 A \sqrt{d}) \geq c
\]
for $A = K/(2 \sqrt{d})$, where $K$ and $c$ are the constants of the small-ball assumption of \Cref{assumption:3}.
\end{proof}

\subsection{Proofs of \Cref{main_result_first_random_dictionary,main_result_second_random_dictionary}} \label{sec:proofs:random:main}
In this section, we only stress the differences between the proofs of \Cref{main_result_first_random_dictionary} and \Cref{main_result_second_random_dictionary} and the results above.

\medskip

\begin{proof}[{Proof of \Cref{main_result_first_random_dictionary}}]
Using \Cref{D_satisfies_NSP_whp}, it follows for the given assumptions, the dictionary $\f{D}$ fulfills the robust NSP of order $s$ with constants $\gamma$ and $\tau$ on an event $E_1$ with $\mathbb{P}(E_1^c) = e^{-\Omega(d)}$. 

So we may proceed as in the deterministic case, apply the small-ball method for $\Phi\mathbf{D}$ and the problem boils down to prove that, for $V:=m^{-1/2}\sum_{i=1}^m \varepsilon_i \varphi_i$, where $\varphi_i$ are as in \Cref{assumption:1} and $\varepsilon_1,\ldots,\varepsilon_m$ independent Rademacher variables,
\begin{equation}
\label{main_problem_paper_random}
    \mathbb{E}\left[\sum_{i=1}^s (\mathbf{D}^{T}V)_i^{*})^2\right]^{1/2} \lesssim \sqrt{\rho s\log\left(\frac{\e n}{s}\right)},
\end{equation}
where $\rho:= \max_{j\in [n]}\mathbb{E}\|\f{d}_j\|_2^2$ is a bound on the maximum of the expected squared column $\ell_2$-norms of $\f{D}$. Importantly, we note that the maximum in $\rho$ is taken outside of the expectation. Furthermore, the expectation is now taken with respect to both $\f{\Phi}$ and $\f{D}$. The constant $\rho$ can be estimated such that
\begin{equation*}
    \rho= \max_{j\in [n]}\mathbb{E}\|{\f{d}_j}\|_2^2 = \max_{j\in [n]} \sum_{i=1}^d \mathbb{E}[\f{d}_{ji}^2] =O(1)
\end{equation*}
since {$\f{d}_{ji} = \psi_{i}[j] / \sqrt{d}$} and $\psi_{i}$ has $\psi_{i}[j]$ has a second moment of order $O(1)$ for any $i \in [n]$ and $j \in [d]$. The proof of \cref{main_problem_paper_random} is then a straightforward modification of the analogous bound for deterministic dictionary, with the main difference that we condition on $\mathbf{D}$ and take the expectation afterwards. Then we can conclude using {Proposition \ref{small_ball_method}}.
\end{proof}

\medskip

\begin{proof}[{Proof of \Cref{main_result_second_random_dictionary}}]
Recall the definition $V:=m^{-1/2}\sum_{i=1}^m \varepsilon_i \varphi_i$ with $\varphi_i$ as in \Cref{assumption:4} and the $\epsilon_i$ as usual. Following the same approach as for \Cref{main_result_first_random_dictionary}, our problem boils down to bound \cref{main_problem_paper_random}. The following proof is similar to \cite[Lemma 6.5]{MendelsonLearning15}.

\begin{proposition}
Consider $\mathbf{D}$ and $\f{\Phi}$ as in \Cref{main_result_second_random_dictionary}. Then,
\begin{equation*}
    \mathbb{E}\left[\sum_{i=1}^s (\mathbf{D}^{T}V)_i^{*})^2\right]^{1/2} \lesssim  \sqrt{s\log\Big(\frac{\e n}{s}\Big)}.
\end{equation*}
\end{proposition}

\begin{proof}
Conditioning on $V$, we obtain, for arbitrary $t > 0$,
\begin{equation} 
    \mathbb{P}\left((\mathbf{D}^T V)_j^{*} \ge t \Big| V\right) \le \binom{n}{j}\mathbb{P}\left(|\langle \f{d}_1^T,V\rangle| > t\sqrt{d} \Big| V \right)^j \le \binom{n}{j} \frac{\|\langle \f{d}_1^T,V\rangle\|_p^{pj}}{(\sqrt{d}t)^{pj}},
\end{equation}
using the independence of the $\f{d}_1,\ldots,\f{d}_n$. Since $V$ is fixed, we apply \Cref{weak_khintchine} for $\mathbf{d}_1$ and write

\begin{equation*}
     \mathbb{P}\left((\mathbf{D}^T V)_j^{*} \ge t|V\right) \lesssim \binom{n}{j} \left(\frac{\sqrt{p}\|V\|_2}{\sqrt{d}t}\right)^{pj}
\end{equation*}
Setting $\sqrt{d}t=u\|V\|_2\sqrt{\log(en/j)}$ and $p = \log(en/j)$, it follows then that
\begin{equation*}
    \mathbb{P}\left((\mathbf{D}^T V)_j^{*}\ge u\frac{\|V\|_2}{\sqrt{d}}\sqrt{\log(en/j)} \bigg| V \right)\le \left(\frac{e}{u}\right)^{j \log(\frac{\e n}{j})}.
\end{equation*}
Integrating the tails we obtain $\mathbb{E}_{\mathbf{D}}(\mathbf{D}^TV)_j^{*})^2 \lesssim \frac{\|V\|_2^2}{d}\log(en/j)$. We now take the expectation both sides with respect to $V$ and get $\mathbb{E}(\mathbf{D}^TV)_j^{*})^2 \lesssim \log(en/j)$ because $\mathbb{E}\|V\|_2^2 = O(d)$. To see this, we compute

{\begin{equation*}
    \mathbb{E}\|V\|_2^2 = \sum_{j=1}^d \mathbb{E} [V_j^2] = \sum_{j=1}^d \frac{1}{m} \left(\sum_{i=1}^m \mathbb{E}[(\varphi_{i}[j])^2] + \sum_{i\neq k} \mathbb{E}\varepsilon_i\varepsilon_k (\varphi_j[i])(\varphi_j[k])\right).
    \end{equation*} }
Now, because of the bounded variance assumption, the first term in the summand is $O(d)$ and by independence of the Rademacher sequence $\varepsilon_i$ the second term in the summand is zero, which shows $\mathbb{E}\|V\|_2^2 = O(d)$.

To finish the proof, just sum over $j\in [s]$, take the square root on both sides and apply Jensen's inequality.
\end{proof}

{The remainder of the proof of \Cref{main_result_second_random_dictionary} is the same as that of \Cref{main_result_first_random_dictionary}}.
\end{proof}

\section{Outlook} \label{sec:conclusion}
In this work we presented an analysis of the $\ell_1$-synthesis method under several, very general assumptions on both the dictionary and measurement matrix, which include in particular heavy-tailed measurements. Specified to the case of a trivial dictionary, i.e., in the standard compressed sensing setting, we showed that a control of only the first $\log(\e n/s)$ moments of the entries of a measurement matrix suffices to establish recovery guarantees for the optimal number of measurements.

It is true that the random assumptions we are assuming for our measurement matrix might not represent directly implementable solutions for many practical applications of compressed sensing, as typically, structural constraints are present. It would be interesting to see whether the presented tools can be helpful to analyze $\ell_1$-synthesis for such settings.

While prior research has shown that $\ell_1$-synthesis can be successful to recover a dictionary-sparse signal $\mathbf{z}_0$ without uniqueness in its sparse representation in the dictionary, we have only focussed on the recovery of both the signal and coefficient vector. It would but interesting to investigate how to establish results on such a non-uniform signal recovery from heavy-tailed measurements in future work.

\section*{Funding}
This work was supported by the National Science Foundation [NSF-IIS-1837991 to C.K.].

\bibliographystyle{alphaabbr}
\bibliography{HeavytailedSynthesisSparsity.bib}

\end{document}